\documentclass[12pt]{article}
\usepackage{amsmath,amssymb,amsthm}
\usepackage{txfonts,bm,pifont}  
\usepackage[dvipdfmx]{graphicx}
\usepackage[dvipdfm,bookmarks=true,bookmarksnumbered=true,colorlinks=true]{hyperref}
\AtBeginDvi{}
\theoremstyle{plain}
\newtheorem{thm}{Theorem}[section]
\newtheorem{prop}[thm]{Proposition}
\newtheorem{defn}[thm]{Definition}

\newtheorem{rem}[thm]{Remark}
\newcommand{\R}{\mathbb{R}}
\newcommand{\Z}{\mathbb{Z}}
\newcommand{\SO}{\mathrm{SO}}
\newcommand{\SU}{\mathrm{SU}}
\newcommand{\tr}{\operatorname{tr}}

\DeclareMathOperator{\sgn}{sgn}
\makeatletter
\@addtoreset{equation}{section} 

\makeatother
\begin{document}
\begin{center}
\begin{Large}
Discrete mKdV and Discrete Sine-Gordon Flows on \\ Discrete Space Curves \\[5mm] 
\end{Large}
\begin{normalsize}
Jun-ichi {\sc Inoguchi}\\
Department of Mathematical Sciences,
Yamagata University\\
Yamagata 990-8560,
Japan\\[2mm] 
Kenji {\sc Kajiwara}\\
Institute of Mathematics for Industry, Kyushu University\\
744 Motooka, Fukuoka 819-0395, Japan\\[2mm]
Nozomu {\sc Matsuura}\\
Department of Applied Mathematics, Fukuoka University\\
Nanakuma 8-19-1, Fukuoka 814-0180 , Japan\\[2mm]
Yasuhiro {\sc Ohta}\\
Department of Mathematics, Kobe University\\
Rokko, Kobe 657-8501, Japan
\end{normalsize}
\end{center}

\begin{abstract}
In this paper, we consider the discrete deformation of the discrete space curves with constant
torsion described by the discrete mKdV or the discrete sine-Gordon equations, and show that it is
formulated as the torsion-preserving equidistant deformation on the osculating plane which satisfies
the isoperimetric condition. The curve is reconstructed from the deformation data by using the
Sym-Tafel formula. The isoperimetric equidistant deformation of the space curves does not preserve
the torsion in general. However, it is possible to construct the torsion-preserving deformation by
tuning the deformation parameters. Further, it is also possible to make an arbitrary choice of the
deformation described by the discrete mKdV equation or by the discrete sine-Gordon equation at each
step. We finally show that the discrete deformation of discrete space curves yields the discrete $K$-surfaces.
\end{abstract}

\section{Introduction}
It is well-known that there are deep connections between the differential geometry and the theory of
the integrable systems, and various integrable differential or difference equations arise as the
compatibility condition of the geometric objects. A typical example is that the surfaces with
constant negative curvature in the Euclidean space (the $K$-surfaces) are described by the
sine-Gordon equation under the Chebychev net parameterization. We refer to \cite{Rogers-Schief:book}
for the detail of such connections.

In accordance with the connection between the differential geometry and the continuous integrable
systems, the studies of {\em discrete differential geometry} started from the mid 1990's in
order to develop its discrete analogue. One of the themes of this area is to construct the geometric
framework, which is consistent with the theory of the discrete integrable systems. Some of the
motivations to study the discrete differential geometry may be, for example, an expectation that
discrete systems may be more fundamental and have rich mathematical structures as was clarified in
the the theory of the integrable systems, or the development of theoretical infrastructure for the
visualization or the simulation of large deformation of the geometric objects. As to the references
to the discrete differential geometry, we refer to \cite{Sauer:book} for an early and embryonic
literature, and to \cite{Bobenko-Suris:book} as a textbook from modern perspective and motivation.

In the deformation theory of the space or plane curves, the Frenet frame and its deformation are
described by the system of linear partial differential equations. The modified KdV (mKdV) or the
nonlinear Schr\"odinger (NLS) equation and their hierarchies arise naturally as the compatibility
condition, as shown by various researchers including Hasimoto and Lamb
\cite{Doliwa-Santini:PLA,Goldstein-Petrich,Hasimoto,Lamb,Langer-Perline,Nakayama_Segur_Wadati:PRL}.
Then in the development of studies of the discrete differential geometry, various continuous
deformations of the discrete curves have been studied. For example, continuous deformations of the
discrete plane and space curves have been considered in
\cite{Doliwa-Santini:JMP,Hisakado-Nakayama-Wadati,Hoffmann:LN,Hoffmann-Kutz,IKMO:JPA}, and
\cite{Doliwa-Santini:JMP,Hisakado-Wadati, Nakayama:JPSJ2007,Nishinari}, respectively, where the
deformations described by the differential-difference analogue of the mKdV and the NLS equations are
formulated. However, the discrete deformation of the discrete curves is not studied well compared to
the continuous deformation. For the plane discrete curves, the isoperimetric deformation described
by the discrete mKdV equation has been studied in \cite{IKMO:KJM,Matsuura:IMRN}. For the 
discrete space curve, the deformation by the discrete sine-Gordon equation have been discussed in
\cite{Doliwa-Santini:dsG}, and the deformation by the discrete NLS equation is formulated in
\cite{Hoffmann:dNLS,Pinkall:dNLS}.

The purpose of this paper is to present the formulation of the discrete analogue of the
isoperimetric deformation by the mKdV equation, which is the most fundamental integrable deformation
of the space curves. In most of the studies on the discrete deformation of the discrete curves
mentioned above, the deformation is described by the Frenet frame, namely, the orthonormal frame
associated with the curve which consists of the tangent, principal normal and the binormal vectors.
This is because the equations for the Frenet frame are nothing but the auxiliary linear problem in
the theory of the integrable systems. However, the deformation of the curves should be described as
the deformation of the position vector.  To this end, we need to `integrate' the deformation equation
for the Frenet frame. However, this procedure is nontrivial for the discrete deformation, and it is
undone in most cases. In \cite{IKMO:KJM,Matsuura:IMRN}, the discrete deformation of the plane
discrete curves described by the discrete mKdV equation has been formulated as the isoperimetric
deformation of the curves.  In this paper, we show that a torsion-preserving isoperimetric and
equidistant deformation for the discrete space curves with constant torsion is described by the
discrete mKdV and the discrete sine-Gordon equations.

This paper is organized as follows. In Section 2, we give a short summary of the $\SO (3)$--$\SU
(2)$ correspondence and that of their Lie algebras $\mathfrak{so}(3)$--$\mathfrak{su}(2)$ which are
frequently used in this paper, for the introduction of notations and the readers' convenience.  We
discuss the torsion-preserving isoperimetric deformation of the space curves with constant torsion
described by the mKdV equation in Section 3. We introduce in Section 4 the discrete space curves and
its Frenet frame, and present the discrete Frenet-Serret formula satisfied by the Frenet frame and
reconstruction of the curve from the Frenet frame by the Sym-Tafel formula. In Section 5, we discuss
the torsion-preserving isoperimetric deformation of the discrete space curves with constant torsion
described the the semi-discrete mKdV equation. In Section 6, we present the torsion-preserving
isoperimetric and equidistant deformation of the discrete space curves with constant torsion which is
the main result of this paper. The proof of the results in Section 6 is given in Section 7.  In
Section 8, we show that the discrete deformation of discrete space curves in Section 6 yields the
discrete $K$-surfaces.

\section{$\SO (3)$ -- $\SU (2)$ Correspondence}
The orthonormal frame of curves in $\R^3$ is given by the matrix in $\SO(3)$.  We sometimes discuss
after transforming it to the matrix in $\SU(2)$ and vice versa. In this section, we give a short
summary of this method\cite{Rogers-Schief:book}.

We choose the basis of $\mathfrak{su} (2)$ as
\begin{equation}\label{basis-su2}
e_1
=
\frac{\sqrt{- 1}}{2}
\left[\begin{array}{cc}
1 & 0\\
0 & - 1
\end{array}\right],\quad
e_2
=
\frac{\sqrt{- 1}}{2}
\left[\begin{array}{cc}
0 & 1\\
1 & 0
\end{array}\right],\quad
e_3
=
\frac{1}{2}
\left[\begin{array}{cc}
0 & - 1\\
1 & 0
\end{array}\right],
\end{equation}
where $e_i$ ($i=1,2,3$) satisfy the commutation relation
\begin{equation}\label{eqn:su2_commutation}
 [e_1,e_2]=e_3,\quad [e_2,e_3]=e_1,\quad [e_3,e_1]=e_2.
\end{equation}
\begin{prop}[$\mathfrak{su}(2)$ -- $\mathbb{R}^3$ correspondence]\label{prop:su(2)_R3}\rm
For $x = {}^{t} \left[x_1, x_2, x_3\right] \in \R^3$, we define the linear map $f\colon \R^3 \to
\mathfrak{su} (2)$ by
\begin{equation}\label{su2_so3:f}
f \left(x\right)
=
x_1 e_1 + x_2 e_2 + x_3 e_3.
\end{equation}
Then $f$ gives an isomorphism of the vector spaces between $\mathfrak{su}(2)$ and $\R^3$.  For
arbitrary $X, Y \in \mathfrak{su} (2)$, we also define the scalar product and the vector product in
$\mathfrak{su} (2)$ by
\begin{equation}
\label{innerproduct-su2}
\langle X, Y\rangle = - 2 \tr \left(X Y\right),\quad
X \times Y = \left[X, Y\right],
\end{equation}
respectively. Then $\mathfrak{su} (2)$ is isomorphic to $\R^3$ as a metric Lie algebra.
\end{prop}
\begin{prop}[$\SU(2)$ -- $\SO(3)$ correspondence]\label{prop:SO(3)_SU(2)}\rm\hfill
\begin{enumerate}
\item For each $\phi \in \mathrm{SU} (2)$ we define a matrix $\varPhi$ by
\begin{equation}\label{eqn:varPhi}
\phi\, f \left(x\right) \phi^{- 1}
=
f \left(\varPhi x\right),\quad
x \in \R^3.
\end{equation}
Then $\varPhi \in \SO (3)$.
\item We write a matrix $\phi\in\mathrm{SU}(2)$ as
\begin{equation}\label{eqn:phi}
\phi
=
\left[\begin{array}{cc}
\alpha & \beta\\
- \beta^\ast & \alpha^\ast
\end{array}\right],\quad
\left|\alpha\right|^2 + \left|\beta\right|^2 = 1,
\end{equation}
where $*$ denotes the complex conjugate, then
$\varPhi\in\mathrm{SO}(3)$ defined in (1) is expressed as
\begin{equation}\label{eqn:varPhi_and_phi}
\varPhi
=
\left[\begin{array}{ccc}
\left|\alpha\right|^2 - \left|\beta\right|^2 &
2 \Re \left(\alpha \beta^\ast\right) &
- 2 \Im \left(\alpha \beta^\ast\right)\\
- 2 \Re \left(\alpha \beta\right) &
\Re \left(\alpha^2 - \beta^2\right) &
- \Im \left(\alpha^2 + \beta^2\right)\\
- 2 \Im \left(\alpha \beta\right) &
\Im \left(\alpha^2 - \beta^2\right) &
\Re \left(\alpha^2 + \beta^2\right)
\end{array}\right].
\end{equation}
\item Conversely, for given $\varPhi = \left[\varPhi_{i j}\right]\in\mathrm{SO}(3)$, the
corresponding $\phi\in\mathrm{SU}(2)$ is determined up to the sign as
\begin{equation}\label{eqn:phi_and_varPhi}
\phi
=
\left[\begin{array}{cc}
\alpha & \beta\\
- \beta^\ast & \alpha^\ast
\end{array}\right],\quad
\left[\begin{array}{c}
\alpha\\
\beta
\end{array}\right]
\;=\;
\pm
\dfrac{1}{2 \sqrt{1 + \tr \varPhi}}
\left[\begin{array}{c}
1 + \tr \varPhi
+ \sqrt{- 1}
\left(\varPhi_{32} - \varPhi_{23}\right)\\
\varPhi_{12} - \varPhi_{21}
+ \sqrt{- 1}
\left(\varPhi_{13} - \varPhi_{31}\right)
\end{array}\right].
\end{equation}
 \end{enumerate}
\end{prop}

There exists an isomorphism between the Lie algebras $\mathfrak{su}(2)$ and $\mathfrak{so}(3)$ which is
consistent with the $\mathrm{SU}(2)$--$\mathrm{SO}(3)$ correspondence in Proposition
\ref{prop:SO(3)_SU(2)}.

\begin{prop}[$\mathfrak{su}(2)$ -- $\mathfrak{so}(3)$ correspondence]\rm 
We define the basis of $\mathfrak{so}(3)$ $E_i$ ($i=1,2,3$) as
\begin{equation}
E_1
=
\left[\begin{array}{ccc}
 0 & 0 & 0\\
0 & 0 & - 1\\
0 & 1 & 0
\end{array}\right],\quad
E_2
=
\left[\begin{array}{ccc}
 0 & 0 & 1\\
0 & 0 & 0\\
- 1 & 0 & 0
\end{array}\right],\quad
E_3
=
\left[\begin{array}{ccc}
 0 & - 1 & 0\\
1 & 0 & 0\\
0 & 0 & 0
\end{array}\right].
\end{equation} 
Then, $\mathfrak{su}(2)$ and $\mathfrak{so} (3)$ are isomorphic by
 the correspondence $e_i\leftrightarrow E_i\, \left(i = 1, 2, 3\right)$.
\end{prop}

\section{Isoperimetric Deformation of Space Curves by mKdV Equation}

Let $\gamma(x,t)$ be a family of the arc-length parameterized space curves.
Here, $x$ is the arc-length at each time $t$. For each $t$, we define the tangent vector 
$T (x, t)$, the principal normal vector $N (x, t)$ and the binormal vector $B (x, t)$ by
\begin{equation}\label{TNB:continuous}
 T(x,t)=\gamma'(x,t),\quad N(x,t)=\frac{\gamma''(x,t)}{|\gamma''(x,t)|},\quad
B(x,t) = T(x,t)\times N(x,t),
\end{equation}
respectively. We also define the curvature $\kappa (x, t)$, the torsion $\lambda (x, t)$ by
\begin{equation}
 \kappa(x,t)=|\gamma''(x,t)|,\quad \lambda(x,t)=-\langle N(x,t),B'(x,t)\rangle,
\end{equation}
respectively. Here $'=\partial/\partial x$，$\dot{}=\partial /\partial t$.
We assume that the torsion is a constant with respect to $x$, namely, 
$\lambda (x, t) = \lambda(t)$, and define the deformation of the curves by
\begin{equation}
\dot{\gamma}
=\label{def:motion}
\left(\frac{\kappa^2}{2} - 3 \lambda^2\right) T
+ \kappa' N
- 2 \lambda \kappa B.
\end{equation}
Then we have the following\cite{Lamb}．
\begin{prop}\label{prop:smooth}\rm\hfill
\begin{enumerate}
\item The arc-length $x$ and the torsion $\lambda$ do not depend on $t$. Namely,
\eqref{def:motion} gives a torsion-preserving isoperimetric deformation of curves.
\item The curvature $\kappa$ satisfies the mKdV equation
\begin{equation}
\dot{\kappa}
=\label{mkdv}
\dfrac{\,3\,}{2} \kappa^2 \kappa'
+ \kappa'''.
\end{equation}
\item The Frenet frame $\varPhi = \left[T, N, B\right]$ satisfies
\begin{align}
& \varPhi'=\varPhi L,\quad 
L=\left[\begin{array}{ccc}
0&-\kappa & 0\\
\kappa & 0 &-\lambda\\
0 & \lambda &0\end{array}\right], \label{Frenet-Serret:continuous}\\
&\dot\varPhi = \varPhi M,\quad
M=\left[
\begin{array}{ccc}
0& - \frac{\kappa^3}{2}+\lambda^2\kappa - \kappa''     & \lambda \kappa'\\
\frac{\kappa^3}{2}-\lambda^2\kappa + \kappa'' & 0 & \lambda\left(-\frac{\kappa^2}{2}+\lambda^2\right)\\
-\lambda \kappa' & \lambda\left(\frac{\kappa^2}{2}-\lambda^2\right) & 0 
\end{array}
\right].
\label{deformation:continuous}
\end{align}
\end{enumerate}
\end{prop}
\begin{proof}
First, we note that $x$ being the arc-length parameter is equivalent to $\langle
T,T\rangle=1$. Differentiation of both sides by $x$ yields $\langle T',T\rangle=0$, which
implies that $T$ and $N$ are orthogonal.  From this and \eqref{TNB:continuous} we have $N\times
B=T$, $B\times T=N$. Moreover, \eqref{Frenet-Serret:continuous} is nothing but the Frenet-Serret
formula, which follows immediately from the definitions of the Frenet frame, the curvature and the
torsion. We show (2) by noticing those notes. Differentiating both sides of $\kappa^2=\langle
T',T'\rangle$, we have by using \eqref{Frenet-Serret:continuous}
\begin{displaymath}
 \dot\kappa \kappa = \langle \dot{T}',T'\rangle = \kappa\langle \dot{T}',N\rangle.
\end{displaymath}
Differentiating both sides of \eqref{def:motion} twice and noticing
\eqref{Frenet-Serret:continuous}, we have
\begin{align}
&\dot{T}=\left(\frac{\kappa^3}{2}-\lambda^2\kappa+\kappa''\right)N - \lambda\kappa' B,\label{dotT}\\
& \dot{T}' = -\kappa\left(\frac{\kappa^3}{2}-\lambda^2\kappa+\kappa''\right)T
+\left(\frac{3}{2}\kappa^2\kappa'+\kappa'''\right)N 
+ \lambda\kappa\left(\frac{\kappa^2}{2}-\lambda^2\right)B,
\label{dotT'}
\end{align}
from which we obtain the mKdV equation for $\kappa$
\begin{displaymath}
\dot\kappa=\langle \dot{T}',N\rangle=\frac{3}{2}\kappa^2\kappa'+\kappa'''.
\end{displaymath}

We next show (1). Independence of the arc-length $x$ from $t$ is equivalent to $\langle
T,T\rangle=1$ for all $t$. Thus it is sufficient to show $\langle T, \dot{T}\rangle=0$ which follows
from differentiation of $\langle T,T\rangle=1$ by $t$, but it follows immediately from \eqref{dotT}.
It is easily shown that the torsion $\lambda$ does not depend on $t$ as follows. Differentiating
both sides of $\lambda=-\langle N,B'\rangle$ by $t$, we have
\begin{equation}\label{dotlambda}
 \dot\lambda=-\langle \dot{N},B'\rangle - \langle N,\dot{B}'\rangle =\lambda\langle \dot{N},N\rangle
- \langle N,\dot{B}'\rangle.
\end{equation}
Then differentiating both sides of $N=\frac{T'}{\kappa}$ by $t$ and rewriting it by using
\eqref{mkdv}, \eqref{Frenet-Serret:continuous} and \eqref{dotT'}, we obtain
\begin{equation}\label{dotN}
 \dot{N}=-\kappa\left(\frac{\kappa^3}{2}-\lambda^2\kappa+\kappa''\right)T 
+ \lambda\kappa\left(\frac{\kappa^2}{2}-\lambda^2\right)B.
\end{equation}
Moreover, differentiating both sides of $B=T\times N$ by $t$ and using \eqref{dotT} and \eqref{dotN}
we get
\begin{equation}\label{dotB}
 \dot{B}=\lambda\kappa'T - \lambda\left(\frac{\kappa^2}{2}-\lambda^2\right)N.
\end{equation}
Further, differentiating both sides of \eqref{dotB} by $x$ and using 
\eqref{Frenet-Serret:continuous} we obtain
\begin{equation}\label{dotB'}
 \dot{B}' = \lambda\left(\frac{\kappa^3}{2}-\lambda^2\kappa+\kappa''\right)T 
- \lambda^2\left(\frac{\kappa^2}{2}-\lambda^2\right)B.
\end{equation}
Substituting \eqref{dotN} and \eqref{dotB'} into \eqref{dotlambda} immediately yields
$\dot{\lambda}=0$. Finally，(3) is derived from \eqref{dotT}, \eqref{dotN} and \eqref{dotB}.
\end{proof}
\begin{rem}\rm
In \cite{Lamb}, the deformation of the curves \eqref{def:motion} is introduced under the assumption
of isoperimetricity ($\dot{x}=0$) and preservation of the constant torsion
($\lambda'=\dot\lambda=0$). Then the mKdV equation \eqref{mkdv} for the curvature and the equations
for the Frenet frame \eqref{Frenet-Serret:continuous}, \eqref{deformation:continuous} are derived
under this assumption. On the other hand, Proposition \ref{prop:smooth} claims that isoperimetricity
and preservation of the constant torsion follow from \eqref{def:motion}.
\end{rem}
If we lift $\varPhi$ to an $\mathrm{SU} (2)$-valued function $\phi$ by using Proposition
\ref{prop:SO(3)_SU(2)}, we see that $\phi$ satisfies
\begin{gather}
\phi'
=\label{akns.mkdv.x}
\phi L,\quad
L
=
\dfrac{1}{2}
\left[\begin{array}{cc}
\sqrt{- 1} \lambda & - \kappa\\
\kappa & - \sqrt{- 1} \lambda
\end{array}\right],\\
\dot{\phi}
=\label{akns.mkdv.t}
\phi M,\quad
M
=
\dfrac{1}{2}
\left[\begin{array}{cc}\smallskip
\sqrt{- 1} \lambda
\left(\frac{\kappa^2}{2} - \lambda^2\right) &
- \frac{\kappa^3}{2} + \lambda^2 \kappa- \kappa'' 
+ \sqrt{- 1} \lambda \kappa'\\
\frac{\kappa^3}{2} - \lambda^2 \kappa + \kappa'' 
+  \sqrt{- 1} \lambda \kappa' &
- \sqrt{- 1} \lambda
\left(\frac{\kappa^2}{2} - \lambda^2\right)
\end{array}\right].
\end{gather}
The integrability condition (compatibility condition) $\dot{L} - M'- \left[L, M\right] = 0$ for the system of
partial differential equations \eqref{Frenet-Serret:continuous}--\eqref{deformation:continuous} or 
\eqref{akns.mkdv.x}--\eqref{akns.mkdv.t} yields the mKdV equation. In particular, 
\eqref{akns.mkdv.x}--\eqref{akns.mkdv.t} coincides with the AKNS representation of the mKdV equation\cite{AKNS}.
Moreover, the torsion $\lambda$ corresponds to the spectral parameter.
\begin{prop}[The Sym-Tafel formula\cite{Sym}]\rm\label{prop:sym:space-curve}
For a solution of the mKdV equation \eqref{mkdv} $\kappa = \kappa \left(x, t\right)$ and
a constant $\lambda$, let $\phi = \phi \left(x, t,\lambda\right)\in\mathrm{SU}(2)$ be a solution of the system of the partial differential equations
\eqref{akns.mkdv.x}--\eqref{akns.mkdv.t}.  For the function $S = \left[S_{i j}\right]\in\mathfrak{su}(2)$ determined by the Sym-Tafel formula
\begin{equation}
S
=\label{sym-tafel}
\left(\dfrac{\partial}{\partial \lambda} \phi\right)
\phi^{- 1}
\end{equation}
we put 
\begin{equation*}
\gamma
=
f^{- 1} \left(S\right)
=
2
\left[\begin{array}{c}
\Im S_{11}\\
\Im S_{21}\\
\Re S_{21}
\end{array}\right]
\end{equation*}
according to the isomorphism $f\colon \R^3 \to \mathfrak{su} (2)$ defined by \eqref{su2_so3:f}.
Then for each $t$, $\gamma$ is a space curve parameterized by the arc-length $x$, and moreover, the
curvature and the torsion are given by $\kappa$ and $\lambda$, respectively. Namely, $\gamma$ gives
the torsion-preserving isoperimetric deformation of the space curves with constant torsion described
by the mKdV equation.
\end{prop}
\begin{proof}
First, we note the isomorphism between $\mathfrak{su}(2)$ and $\R^3$ given in Proposition \ref{prop:su(2)_R3}.
Since we have 
\begin{equation*}
S'
=
\phi \left(\dfrac{\partial}{\partial \lambda} L\right) \phi^{- 1}
=
\phi\, e_1 \phi^{- 1}
\end{equation*}
by differentiating the Sym-Tafel formula \eqref{sym-tafel} by $x$, 
we see $\left|S'\right| = 1$. Putting $T=S'$ and differentiating by $x$ once more, we get
$T'= \phi \left[L, e_1\right] \phi^{- 1}= \kappa\, \phi e_2 \phi^{- 1}$.
So taking the Frenet-Serret formula (\ref{Frenet-Serret:continuous}) into account, we put
$N = \phi e_2 \phi^{- 1}$. From $N'= - \kappa T + \lambda\, \phi e_3 \phi^{- 1}$
we also put $B = \phi e_3 \phi^{- 1}$. Then it holds that $T \times N = B$ and $B' = - \lambda N$.
Therefore we have shown that $\left[T, N, B\right]$ satisfies the Frenet-Serret formula.
Next, differentiating the Sym-Tafel formula \eqref{sym-tafel} by $t$, we have
\begin{equation*}
\dot{S}
=
\phi \left(\dfrac{\partial}{\partial \lambda} M\right)
\phi^{- 1}
=
\phi
\left\{\left(\dfrac{\kappa^2}{2} - 3 \lambda^2\right) e_1
+ \kappa' e_2
- 2 \lambda \kappa e_3\right\}
\phi^{- 1},
\end{equation*}
which coincides with the definition of the isoperimetric deformation 
\eqref{def:motion}.
\end{proof}
Suppose that we are going to reconstruct the isoperimetric deformation of space curves from the
specified values of the curvature and the constant torsion. Since the curvature is determined from
the second derivative of the curve, it is necessary to integrate twice in order to reconstruct the
curves. The Sym-Tafel formula \eqref{sym-tafel} claims that if we know all of the entries of the
matrix $\phi$ explicitly (namely, first integration has been performed by a certain method and we
have the explicit form of the tangent vector), then the position vectors of curves at each time are
obtained without the second integration.
\section{Discrete Space Curve}
\label{section:space_discrete_curve}
In this section, we introduce the discrete space curve and its Frenet frame, and discuss the
Frenet-Serret formula and the Sym-Tafel formula.

\begin{defn}[\cite{Eyring,Sauer:book}]\rm\hfill
\begin{enumerate}
 \item For a map $\gamma\colon \Z \to \R^3,\,n \mapsto \gamma_n$, if any consecutive three points
$\gamma_{n +1}$, $\gamma_n$ and $\gamma_{n - 1}$ are not colinear, then 
we call $\gamma$ a {\em discrete space curve}.
 \item For a discrete space curve $\gamma$ we set
\begin{equation}
\epsilon_n = \left|\gamma_{n + 1} - \gamma_n\right|
\end{equation}
and introduce
\begin{equation}
T_n = \dfrac{\gamma_{n + 1} - \gamma_n}{\epsilon_n},\quad
N_n = B_n \times T_n,\quad
B_n = \dfrac{T_{n - 1} \times T_n}{\left|T_{n - 1} \times T_n\right|}
\end{equation}
which we call the {\em tangent vector}, the {\em (principal) normal vector} and the {\em binormal vector} respectively.
We also call the matrix valued function $\varPhi = \left[T, N, B\right]\colon \Z \to \SO (3)$ the {\em Frenet frame}
of $\gamma$.
\end{enumerate}
\end{defn}
By definition, we see that the Frenet frame $\varPhi_n=[T_n,N_n,B_n]$ satisfies the following difference equation:
\begin{equation}
\varPhi_{n + 1}
=\label{eq:frenet-serret:d.space-curve}
\varPhi_n
R_1 \left(- \nu_{n + 1}\right)
R_3 \left(\kappa_{n + 1}\right),
\end{equation}
where $R_1, R_3$ are rotation matrix given by
\begin{equation}\label{eqn:rotation_matrix}
R_1 \left(\theta\right)
=
\left[\begin{array}{ccc}
1 & 0 & 0\\
0 & \cos \theta & - \sin \theta\\
0 & \sin \theta & \cos \theta
\end{array}\right],\quad
R_3 \left(\theta\right)
=
\left[\begin{array}{ccc}
\cos \theta & - \sin \theta & 0\\
\sin \theta & \cos \theta & 0\\
0 & 0 & 1
\end{array}\right],
\end{equation}
and $\kappa\colon \Z \to \left(0, \pi\right)$, 
$\nu\colon \Z \to \left[- \pi, \pi\right)$ are the angles defined by
\begin{equation}\label{kappa_and_nu:discrete}
\langle T_n, T_{n - 1}\rangle = \cos \kappa_n,\quad
\langle B_n, B_{n - 1}\rangle = \cos \nu_n,\quad
\langle B_n, N_{n - 1}\rangle = \sin \nu_n,
\end{equation}
respectively. We call \eqref{eq:frenet-serret:d.space-curve}--\eqref{kappa_and_nu:discrete} the
Frenet-Serret formula in the same manner as the space smooth curve.

If we lift the Frenet frame $\varPhi$ to an $\SU (2)$-valued function $\phi$ according to
Proposition \ref{prop:SO(3)_SU(2)}, then $\phi$ satisfies
\begin{gather}
\phi_{n + 1}
=\label{eq:frenet-serret:d.space-curve:su2}
\phi_n L_n,\\
L_n
=\label{eq:frenet-serret:d.space-curve:su2-L}
\pm
\left[\begin{array}{cc}\medskip
e^{-\frac{\sqrt{-1}}{2}\nu_{n+1}}\cos \dfrac{\kappa_{n + 1}}{2} 
& -e^{-\frac{\sqrt{-1}}{2}\nu_{n+1}}\sin \dfrac{\kappa_{n + 1}}{2}\\
e^{\frac{\sqrt{-1}}{2}\nu_{n+1}}\sin \dfrac{\kappa_{n + 1}}{2} 
&e^{\frac{\sqrt{-1}}{2}\nu_{n+1}}\cos \dfrac{\kappa_{n + 1}}{2}
\end{array}\right].
\end{gather}
\begin{defn}\rm
We call the function defined by
\begin{equation}
\lambda_n = \dfrac{\sin \nu_{n+1}}{\epsilon_n}\label{torsion} 
\end{equation}
the {\em torsion} of the discrete space curve.
\end{defn}
In the following, we consider the discrete space curve whose torsion is a constant
\begin{equation}\label{const-torsion}
 \lambda_n = \lambda.
\end{equation}
Then, the following proposition holds.
\begin{prop}[The Sym-Tafel formula]\rm
For $\lambda\in\R$ and functions $\nu_n\in[-\pi,\pi)$, $\kappa_n\in(0, \pi)$, let
$\phi_n\in\mathrm{SU}(2)$ be a solution of the difference equations
\eqref{eq:frenet-serret:d.space-curve:su2}-- \eqref{const-torsion}.  Let $f$ be the isomorphism
defined by \eqref{su2_so3:f} and set
\begin{equation}\label{S:discrete}
\gamma_n
=
f^{- 1} \left(S_n\right),\quad
S_n
=
- \left(\dfrac{\partial}{\partial \lambda} \phi_n\right)
{\phi_n}^{- 1}.
\end{equation}
Then $\gamma$ satisfies the Frenet-Serret formula
\eqref{eq:frenet-serret:d.space-curve}--\eqref{kappa_and_nu:discrete}. Namely, $\gamma$ is the 
discrete space curve with the constant torsion $\lambda$, and the distance between the contiguous
points $\gamma_n$ and $\gamma_{n+1}$ is $\epsilon_n$.  Moreover, the angle between the
contiguous tangent vectors $T_{n-1}$ and $T_{n}$ is given by $\kappa_n$, and the angle between the
contiguous binormal vectors $B_{n-1}$ and $B_{n}$ is  $\nu_n$.
\end{prop}
\begin{proof}
Taking the difference of $S$ with respect to $n$, we have from 
\eqref{eq:frenet-serret:d.space-curve:su2-L}--\eqref{S:discrete}
$S_{n + 1} - S_n = - \phi_n \left(L_n\right)_\lambda L_n^{- 1} \phi_n^{- 1} = \epsilon_n \phi_n e_1 \phi_n^{- 1}$. Then we put
\begin{equation}\label{Sym_T}
T_n =\frac{S_{n +1} - S_n}{\epsilon_n} =  \phi_n e_1 \phi_n^{- 1}.
\end{equation}
Since we have
\begin{align*}
T_{n+1}
&=\phi_{n+1} e_1 {\phi_{n+1}}^{-1}
=\phi_{n}L_n e_1 {L_n}^{-1}{\phi_{n}}^{-1}\\
&=\cos \kappa_{n+1} T_n 
+ \cos \nu_{n+1}\sin\kappa_{n+1} \phi_n e_2 {\phi_n}^{-1}
- \sin \nu_{n+1}\sin\kappa_{n+1}\phi_n e_3 {\phi_n}^{-1},
\end{align*}
we also put 
\begin{equation}\label{Sym_NB}
N_n =\phi_n e_2 {\phi_n}^{-1},\quad
B_n =\phi_n e_3 {\phi_n}^{-1}.
\end{equation}
Further, $N_{n+1}$, $B_{n+1}$ can be obtained as
\begin{align*}
& N_{n+1} = -\sin\kappa_{n+1} T_n +\cos \nu_{n+1}\cos
 \kappa_{n+1} N_n - \sin\nu_{n+1}\cos \kappa_{n+1} B_n\, ,\\
& B_{n+1} = \sin\nu_{n+1}N_n + \cos\nu_{n+1}B_n\, ,
\end{align*}
by the similar calculation, from which we see that
$\left[T_n, N_n, B_n\right]$ satisfies the Frenet-Serret formula.
\end{proof}
For the later convenience of the notation, we put
\begin{equation}\label{eqn:a_and_epsilon_nu}
 a_n = \left(1+\tan^2\frac{\nu_{n+1}}{2}\right)\epsilon_n.
\end{equation}
Then we have
\begin{align}
 & \epsilon_n = \frac{a_n}{1+\frac{a_n^2\lambda^2}{4}},\label{epsilon}\\
 & \nu_{n+1} = 2\arctan\frac{a_n\lambda}{2}.
\end{align}

\section{Continuous Isoperimetric Deformation of Discrete Space Curves}\label{section:semi-discrete}
In this section, we consider a continuous isoperimetric deformation for the discrete space curves
with a constant torsion described by the semi-discrete mKdV equation. We consider a family of
discrete space curves $\gamma_n (t)$ with the parameter $t$, and define $\epsilon_n (t)$, $\kappa_n
(t)$, $\nu_n (t)$ and $\lambda_n(t)$ as given in Section \ref{section:space_discrete_curve}. Here we
assume that $\epsilon_n(t)=\epsilon(t)$ and $\nu_n(t)=\nu(t)$ do not depend on $n$. Then $a_n(t)=a(t)$ is independent of
$n$. Also, the torsion $\lambda(t)$ of each curve $\gamma$ is a constant with respect to $n$. In the
following we omit writing $t$ explicitly.  Now we determine the direction of the deformation
$\dot{\gamma}_n$ of each vertex of the discrete curves as
\begin{equation}\label{def:motion-sd}
\dot{\gamma}_n
=
\dfrac{\,\epsilon\,}{a}
\left(\cos \nu T_n
- \cos \nu \tan \frac{\kappa_n}{2} N_n
+ \sin \nu \tan \frac{\kappa_n}{2} B_n\right).
\end{equation}
\begin{prop}\rm
Suppose that a family of discrete space curves $\gamma_n(t)$ is deformed according to (\ref{def:motion-sd}).
Then we have the following:
\begin{enumerate}
 \item $\epsilon$, $a$, $\nu$ and $\lambda$ are constants with respect to $t$.
 \item $\kappa$ varies according to the semi-discrete mKdV equation
\begin{equation}
\dot{\kappa}_n
=\label{sdmkdv}
\dfrac{\,1\,}{a}
\left(\tan \dfrac{\kappa_{n + 1}}{2}
- \tan \dfrac{\kappa_{n - 1}}{2}\right).
\end{equation}
 \item The Frenet frame $\varPhi = \left[T, N, B\right]$ satisfies the following system of equations.
\begin{align}
& \varPhi_{n+1} = \varPhi_n L_n,\quad L_n = 
R_1 \left(- \nu\right)
R_3 \left(\kappa_{n + 1}\right),\label{eq:frenet-serret:d.space-curve_SO(3)}\\
&\dot{\varPhi}_{n} = \varPhi_n M_n,\ 
M_n=\frac{1}{a}
 \left[\begin{array}{ccc}\smallskip
0&-\cos\nu\tan\frac{\kappa_n}{2} - \tan\frac{\kappa_{n+1}}{2} & \sin\nu\tan\frac{\kappa_n}{2}\\
\smallskip
\cos\nu\tan\frac{\kappa_n}{2} + \tan\frac{\kappa_{n+1}}{2} & 0 & \sin \nu\\
-\sin\nu\tan\frac{\kappa_n}{2} & -\sin \nu & 0
\end{array}\right].  \label{eq:sdmKdV:M_SO(3)} 
\end{align}
\end{enumerate}
\end{prop}
\begin{proof}
We note that \eqref{eq:frenet-serret:d.space-curve_SO(3)} is nothing but the Frenet-Serret formula 
\eqref{eq:frenet-serret:d.space-curve}, which follows directly from the definition of the discrete space curve.
We first show (2). Differentiating $\cos\kappa_n = \langle T_n,T_{n-1}\rangle$ (see \eqref{kappa_and_nu:discrete}) by $t$ yields
\begin{equation}\label{dot_kappa}
 -\dot\kappa_n\sin\kappa_n = \langle \dot T_n,T_{n-1}\rangle + \langle T_n,\dot T_{n-1}\rangle.
\end{equation}
Noticing that \eqref{def:motion-sd} is expressed as
\begin{displaymath}
 \dot\gamma_n = \frac{\epsilon}{a}\varPhi_n
\left[
\begin{array}{c}\cos\nu \\ -\cos\nu\tan\frac{\kappa_n}{2}\\ \sin\nu\tan\frac{\kappa_n}{2}\end{array}
\right],
\end{displaymath}
we have by using \eqref{eq:frenet-serret:d.space-curve_SO(3)} 
\begin{align}
\dot\gamma_{n+1} - \dot\gamma_n 
& = \frac{\epsilon}{a}\varPhi_n
\left\{
L_n
\left[
\begin{array}{c}\cos\nu \\ -\cos\nu\tan\frac{\kappa_{n+1}}{2}\\ \sin\nu\tan\frac{\kappa_{n+1}}{2}\end{array}
\right]
-
\left[
\begin{array}{c}\cos\nu \\ -\cos\nu\tan\frac{\kappa_n}{2}\\ \sin\nu\tan\frac{\kappa_n}{2}\end{array}
\right]\right\}\\
&=\frac{\epsilon}{a}\varPhi_n
\left[
\begin{array}{c}
 0 \\
 \cos\nu\tan\frac{\kappa_n}{2}+\tan\frac{\kappa_{n+1}}{2}\\
 -\sin\nu\tan\frac{\kappa_n}{2}
\end{array}
\right],
\label{dot_T_n_1}
\end{align}
or
\begin{equation}
 \dot T_n = \frac{1}{a}\varPhi_n
\left[
\begin{array}{c} 0 \\ \cos\nu\tan\frac{\kappa_n}{2}+\tan\frac{\kappa_{n+1}}{2}\\ -\sin\nu\tan\frac{\kappa_n}{2}\end{array}
\right].
\label{dot_T_n}
\end{equation}
Also, from \eqref{dot_T_n},
\begin{equation}\label{T_n-1}
 T_{n-1}=\varPhi_nL_{n-1}^{-1}
\left[\begin{array}{c} 1 \\ 0 \\ 0 \end{array}\right]
=\varPhi_n
\left[\begin{array}{c} \cos\kappa_n \\ -\sin\kappa_n \\ 0 \end{array}\right],
\end{equation}
and 
\begin{align}
\dot T_{n-1} &=\frac{1}{a}\varPhi_{n-1}
\left[
\begin{array}{c} 0 \\ \cos\nu\tan\frac{\kappa_{n-1}}{2}+\tan\frac{\kappa_{n}}{2}\\ -\sin\nu\tan\frac{\kappa_{n-1}}{2}\end{array}
\right]
=\frac{1}{a}\varPhi_{n}L_{n-1}^{-1}
\left[
\begin{array}{c} 0 \\ \cos\nu\tan\frac{\kappa_{n-1}}{2}+\tan\frac{\kappa_{n}}{2}\\ -\sin\nu\tan\frac{\kappa_{n-1}}{2}\end{array}
\right]\nonumber\\
&=\frac{1}{a}\varPhi_{n}
\left[
\begin{array}{c}
 \sin\kappa_n\tan\frac{\kappa_{n-1}}{2} + \cos\nu\sin\kappa_n\tan\frac{\kappa_{n}}{2}\\
\cos\kappa_n\tan\frac{\kappa_{n-1}}{2} +\cos\nu\cos\kappa_n\tan\frac{\kappa_{n}}{2}\\
\sin\nu\tan\frac{\kappa_n}{2}  \end{array}
\right], \label{dot_T_n-1}
\end{align}
we obtain
\begin{align*}
 \langle \dot T_n,T_{n-1}\rangle + \langle T_n,\dot T_{n-1}\rangle
=\frac{1}{a}\sin\kappa_n\left(\tan\frac{\kappa_{n-1}}{2} - \tan\frac{\kappa_{n+1}}{2}\right).
\end{align*}
Therefore the semi-discrete mKdV equation \eqref{sdmkdv} for $\kappa_n$ is derived immediately from
\eqref{dot_kappa}. We next show (1). Differentiating $\epsilon^2=\langle \gamma_{n+1}-\gamma_n,\gamma_{n+1}-\gamma_n\rangle$
by $t$ and noticing \eqref{dot_T_n_1}, we have
\begin{align*}
2\epsilon\dot\epsilon&=2\langle \dot\gamma_{n+1}-\dot\gamma_n,\gamma_{n+1}-\gamma_n\rangle 
=2\epsilon\langle \dot\gamma_{n+1}-\dot\gamma_n,T_n\rangle = 0,
\end{align*}
which implies $\dot\epsilon=0$.  In order to show $\dot\nu=0$, we
differentiate both sides of $\cos\nu=\langle B_n,B_{n+1}\rangle$ to get
\begin{equation}
 -\dot\nu\sin\nu = \langle \dot B_{n},B_{n+1}\rangle + \langle B_{n}, \dot B_{n+1}\rangle .\label{dot_nu}
\end{equation}
Since $B_n=\frac{T_{n-1}\times T_n}{\left|T_{n-1}\times T_n\right|}
=\frac{T_{n-1}\times T_n}{\sin\kappa_n}$, we have
\begin{displaymath}
 \dot B_n = \frac{d}{dt}\left(\frac{1}{\sin\kappa_n}\right)T_{n-1}\times T_n
+ \frac{1}{\sin\kappa_n}\dot T_{n-1}\times T_n 
+ \frac{1}{\sin\kappa_n} T_{n-1}\times \dot T_n .
\end{displaymath}
By using \eqref{dot_T_n}--\eqref{dot_T_n-1} and noticing
\begin{equation}
 T_n\times N_n = B_n,\quad N_n\times B_n = T_n,\quad B_n\times T_n=N_n,
\end{equation}
we obtain
\begin{align}
 \dot B_n& =\frac{1}{a}\varPhi_n
\left[
\begin{array}{c}
\sin\nu\tan\frac{\kappa_n}{2}\\
\sin\nu\\
\frac{a}{\tan\kappa_n}\left\{-\dot\kappa_n + \frac{1}{a}\left(\tan\frac{\kappa_{n+1}}{2} - \tan\frac{\kappa_{n-1}}{2}\right)\right\}\end{array}
\right] \label{dot_B_n_0}\\
& =\frac{1}{a}\varPhi_n
\left[
\begin{array}{c}
\sin\nu\tan\frac{\kappa_n}{2}\\
\sin\nu\\
0
\end{array}
\right].
\label{dot_B_n}
\end{align}
Further, since we see that
\begin{equation}
 B_{n+1}=\varPhi_n L_n
\left[
\begin{array}{c}
0\\
0\\
1
\end{array}
\right]
=\varPhi_n
\left[
\begin{array}{c}
0\\
\sin\nu\\
\cos\nu
\end{array}
\right],
\label{B_n+1}
\end{equation}
\begin{align}
 \dot B_{n+1} &= \frac{1}{a}\varPhi_n L_n
\left[
\begin{array}{c}
\sin\nu\tan\frac{\kappa_{n+1}}{2}\\
\sin\nu\\
0
\end{array}
\right]
=
\frac{1}{a}\varPhi_n 
\left[
\begin{array}{c}
-\sin\nu\tan\frac{\kappa_{n+1}}{2}\\
\sin\nu\cos\nu\\
-\sin^2\nu
\end{array}
\right],
\end{align}
we obtain from \eqref{dot_nu}
\begin{displaymath}
 -\dot\nu \sin\nu
= \langle \dot B_{n},B_{n+1}\rangle + \langle B_{n}, \dot B_{n+1}\rangle
= \frac{\sin^2\nu}{a} - \frac{\sin^2\nu}{a} =0,
\end{displaymath}
which proves $\dot\nu=0$. Finally, (3) can be shown from \eqref{dot_T_n}, \eqref{dot_B_n} and that $M_n\in\mathfrak{so}(3)$.
\end{proof}
\begin{rem}\rm
The integrability condition $\dot{L}_n - L_n M_{n + 1} + M_n L_n = 0$ of the system of differential
and difference equations \eqref{eq:frenet-serret:d.space-curve_SO(3)}, \eqref{eq:sdmKdV:M_SO(3)}
yields the semi-discrete mKdV equation \eqref{sdmkdv}.  Moreover in \eqref{dot_B_n_0}, since
$M_n\in\mathfrak{so}(3)$, if we express $\dot B_n$ by the linear combination of $T_n$, $N_n$, $B_n$,
the coefficient of $B_n$ must be $0$. The semi-discrete mKdV equation \eqref{sdmkdv} also follows
from this condition.
\end{rem}

The Sym-Tafel formula acts as the formula for reconstruction of the discrete curves.  If we
lift $\varPhi= \left[T, N, B\right]$ to an $\mathrm{SU} (2)$-valued function $\phi$, then
$\phi$ satisfies
\begin{align}
&\phi_{n + 1}
=\label{akns.sdmkdv.n}
\phi_n L_n,\quad
L_n
=
\pm
\left[\begin{array}{cc}
e^{-\frac{\sqrt{- 1}}{2} \nu}
\cos \dfrac{\kappa_{n + 1}}{2} &
-e^{-\frac{\sqrt{- 1}}{2} \nu}
\sin \dfrac{\kappa_{n + 1}}{2}\\
e^{\frac{\sqrt{- 1}}{2} \nu}
\sin \dfrac{\kappa_{n + 1}}{2} &
e^{\frac{\sqrt{- 1}}{2} \nu}
\cos \dfrac{\kappa_{n + 1}}{2}
\end{array}\right]
,\\
&\dot{\phi}_n
=\label{akns.sdmkdv.t}
\phi_n M_n,\quad
M_n
=
\dfrac{1}{2 a}
\left[\begin{array}{cc}
- \sqrt{- 1} \sin \nu &
- e^{- \sqrt{- 1} \nu}
\tan \dfrac{\kappa_n}{2}
- \tan \dfrac{\kappa_{n + 1}}{2}\\
e^{\sqrt{- 1} \nu}
\tan \dfrac{\kappa_n}{2}
+ \tan \dfrac{\kappa_{n + 1}}{2} &
\sqrt{- 1} \sin \nu
\end{array}\right].
\end{align}
The following proposition is a consequence of application of the Sym-Tafel formula to the system of
differential and difference equations \eqref{akns.sdmkdv.n} and \eqref{akns.sdmkdv.t}.
\begin{prop}\label{prop:sym:sd.space-curve}\rm
Let $a > 0$. For a solution $\kappa = \kappa_n (t)$ of the semi-discrete mKdV equation
\eqref{sdmkdv} and a constant $\nu$, let $\phi = \phi_n (t)$ be a solution to the system of the
differential and difference equations \eqref{akns.sdmkdv.n}--\eqref{akns.sdmkdv.t}. Moreover, let
$f$ be the isomorphism defined in \eqref{su2_so3:f}, and let
\begin{equation}
\gamma_n
=
f^{- 1} \left(S_n\right),\quad
S_n
=
- \left(\dfrac{\partial}{\partial \lambda} \phi_n\right)
{\phi_n}^{- 1},\quad
\lambda = \dfrac{2}{a} \tan \dfrac{\nu}{2}.
\end{equation}
Then, for each $t$, $\gamma$ is a discrete curve of which the distance between the contiguous vertices
is a constant
\begin{equation}
\epsilon
=
\dfrac{a}{1+ \frac{a^2\lambda^2}{4}},
\end{equation}
the angle between the contiguous tangent vectors is $\kappa$, and the angle between the contiguous
binormal vectors is $\nu$. Namely, $\gamma$ is the torsion-preserving isoperimetric deformation of
the discrete space curves with constant torsion described by the semi-discrete mKdV equation.
\end{prop}
\begin{proof}
Differentiating $S$ by $t$ yields
\begin{equation*}
\dot{S}_n
\;=\;
- \phi_n \left(\dfrac{\partial}{\partial \lambda} M_n\right)
\phi_n^{- 1}
\;=\;
\phi_n
\left(\cos \nu\, e_1
- \cos \nu \tan \frac{\kappa_n}{2} e_2
+ \sin \nu \tan \frac{\kappa_n}{2} e_3\right)
\phi_n^{- 1},
\end{equation*}
which coincides with the definition of the isoperimetric deformation \eqref{def:motion-sd} from \eqref{Sym_T}, \eqref{Sym_NB}.
\end{proof}
The above result implies that it is possible to obtain the position vectors of the discrete space
curves without integration (summation), if we have the entries of the matrix $\phi$ explicitly.
\begin{rem}\rm\hfill
\begin{enumerate}
\item Doliwa and Santini\cite{Doliwa-Santini:PLA,Doliwa-Santini:JMP,Doliwa-Santini:dsG} discussed
the deformation of the curves restricted on the sphere $S^{N-1}(r)$ of radius $r$ in
$\mathbb{R}^N$. In particular, for the case of $N=3$, they derived the deformation described by the
semi-discrete mKdV equation \eqref{sdmkdv} as one of the isoperimetric deformations of the discrete
curves. It is known that the smooth curves on $S^2(\frac{1}{|\lambda|})$ is equivalent to the smooth
curves with constant torsion in $\mathbb{R}^3$, and the explicit correspondence between them is
given as follows.
Suppose that $x$ is the arc-length and $\Gamma(x)$ is a curve on $S^2(\frac{1}{|\lambda|})$, then
$\gamma(x)=\lambda\int \Gamma(x)\times \Gamma'(x)dx$ is a curve with the torsion $\lambda$ in $\mathbb{R}^3$. 
Conversely, suppose that $\gamma(x)$ is a curve with the torsion $\lambda$ in $\mathbb{R}^3$, and 
$B(x)$ is the binormal vector of $\gamma(x)$, then 
$\Gamma(x)=\pm \frac{1}{\lambda}B(x)$ is a curve on $S^2(\frac{1}{|\lambda|})$\cite{Bates-Melko,Koenigs}.

A similar correspondence exists for the discrete curves.
In fact, suppose that $\Gamma_n$ is a discrete curve (with a certain condition) on $S^2(\frac{1}{|\lambda|})$, then 
\begin{equation}\label{gamma_Gamma:discrete}
\gamma_n=\lambda\sum_{k}^n\Gamma_{k-1}\times \Gamma_k
\end{equation}
is a discrete curve with the constant torsion $\lambda$ in $\mathbb{R}^3$, and conversely, suppose
that $\gamma_n$ is a discrete curve with the constant torsion $\lambda$ in $\mathbb{R}^3$ and $B_n$ 
is the binormal vector of $\gamma_n$, 
\begin{equation}\label{Gamma_gamma:discrete}
\Gamma_n = \pm\frac{1}{\lambda} B_n
\end{equation}
is a discrete curve on $S^2(\frac{1}{|\lambda|})$. We refer to Appendix \ref{section:app_A} for the
details.  Note that it is necessary to perform summation to derive the isoperimetric deformation
\eqref{def:motion-sd} of $\gamma_n$ from the isoperimetric deformation of $\Gamma_n$ described by
the semi-discrete mKdV equation in \cite{Doliwa-Santini:JMP}, as is seen from
\eqref{gamma_Gamma:discrete}.
\item Generalizing the system of equations \eqref{eq:frenet-serret:d.space-curve_SO(3)},
\eqref{eq:sdmKdV:M_SO(3)} by introducing appropriate effect of bend and twist, it is possible to
       construct a discrete model of elastic curve\cite{Nishinari}．
\end{enumerate}
\end{rem}

\section{Discrete Isoperimetric Deformation of Discrete Space Curves}

In this section we discuss the discrete isoperimetric deformation of the discrete space curve
$\gamma$.  We write the deformed curve as $\overline{\gamma}$, and we also express the data such as the
Frenet frame or the torsion associated with $\overline{\gamma}$ by putting
$\overline{\phantom{\epsilon}}$. For example, for the data $\epsilon_n$, $\nu_n$, $\kappa_n$,
$\lambda_n$, $a_n$ associated with $\gamma$, the corresponding data for $\overline{\gamma}$ are
denoted as $\overline{\epsilon}_n$, $\overline{\nu}_n$, $\overline{\kappa}_n$,
$\overline{\lambda}_n$, $\overline{a}_n$, respectively. Now we assume that the torsion of $\gamma$
is a constant, namely,
\begin{equation}
\dfrac{2}{a_n} \tan \dfrac{\nu_{n + 1}}{2}
=\lambda\quad (\text{const.})~.
\end{equation}
Then we have the following proposition.
\begin{prop}\label{prop:dmkdvflow}\rm
For a discrete space curve $\gamma$ with the constant torsion $\lambda$, we define the
new discrete space curve $\overline{\gamma}$ by
\begin{gather}
\overline{\gamma}_n
=\label{def:isoperi0}
\gamma_n
+
\delta \left(\cos w_n T_n + \sin w_n N_n\right),\\
\delta =\label{def:isoperi-c0}
\dfrac{b}{1 + \frac{b^2\lambda^2}{4}}>0, \\
w_{n + 1} 
= \label{def:isoperi-w0}
- \kappa_{n + 1} + 2\arctan\dfrac{b + a_n}{b - a_n}\tan \dfrac{w_n}{2}.
\end{gather}
Here we choose the constants $b>0$ and $w_0$ such that the sign of
$\sin(w_{n+1}+\kappa_{n+1}-w_{n-1})$ is constant for all $n$. Then we have:
\begin{enumerate}
 \item (Isoperimetricity) $\overline{\gamma}$ satisfies
\begin{equation}
\overline{\epsilon}_n=\left|\overline{\gamma}_{n + 1} - \overline{\gamma}_n\right|
= \left|\gamma_{n + 1} - \gamma_n\right|
= \epsilon_n,
\end{equation}
and thus the deformation $\gamma \mapsto \overline{\gamma}$ defined by
\eqref{def:isoperi0} is an isoperimetric deformation.
 \item (Preservation of the torsion) It follows that
\begin{equation}
 \overline{\nu}_{n}=\nu_n.
\end{equation}
Therefore, we have that $\overline{a}_n=a_n$ from (1), and that the torsion 
$\overline{\lambda}_n=\frac{\sin\overline{\nu}_{n+1}}{\overline{\epsilon}_n}$ of
$\overline{\gamma}$ is a constant with respect to $n$, whose value is equal to $\lambda$. Namely,
the deformation $\gamma \mapsto \overline{\gamma}$ preserves the torsion.
\item (Deformation of the Frenet frame) The Frenet frame 
$\overline{\varPhi} = \left[\overline{T},\overline{N}, \overline{B}\right]$ of 
$\overline{\gamma}$ satisfies either of the following:
in case of $\sin(w_{n+1}+\kappa_{n+1}-w_{n-1})>0$, 
\begin{equation}
\overline{\varPhi}_n
=
\varPhi_n
R_3 \left(w_n\right)
R_1 \left(\mu\right)
R_3 \left(w_{n + 1} + \kappa_{n + 1}\right),\quad
\mu = - 2 \arctan \dfrac{b\lambda}{2},
\end{equation}
or in case of $\sin(w_{n+1}+\kappa_{n+1}-w_{n-1})<0$,
\begin{equation}
\overline{\varPhi}_n
=
\varPhi_n
R_3 \left(w_n\right)
R_1 \left(\mu\right)
R_3 \left(-w_{n + 1} - \kappa_{n + 1}\right),\quad
\mu =  2 \arctan \dfrac{2}{b\lambda}.
\end{equation}
\end{enumerate}
\end{prop}
The proof of Proposition \ref{prop:dmkdvflow} will be given in Section \ref{subsection_main1}.  The
procedure to obtain $\overline{\gamma}$ from $\gamma$ can be also expressed as follows: (1) Take an
arbitrary point on $\gamma$, for example, $\gamma_0$, and move it to an arbitrary point on the plane
(the osculating plane) spanned by $T_0$ and $N_0$ (or $T_{-1}$ and $T_0$), and we denote the
distance that the point has moved by $\delta$. Here, for the torsion $\lambda$ of $\gamma$, $\delta$
is chosen so that $0<\delta<1/\lambda$ is satisfied. (2) Draw a circle with the radius $\delta$ on
the osculating plane of each point. (3) Move other points to the position satisfying the following
three conditions; (a) it is on the circle (equidistant condition), (b) it preserves the arc-length
(isoperimetric condition) (c) it lies on the lower half plane of the osculating plane (opposite side
of $N$ with respect to $T$).
\begin{figure}[ht]
\begin{center}
  \includegraphics[scale=0.25]{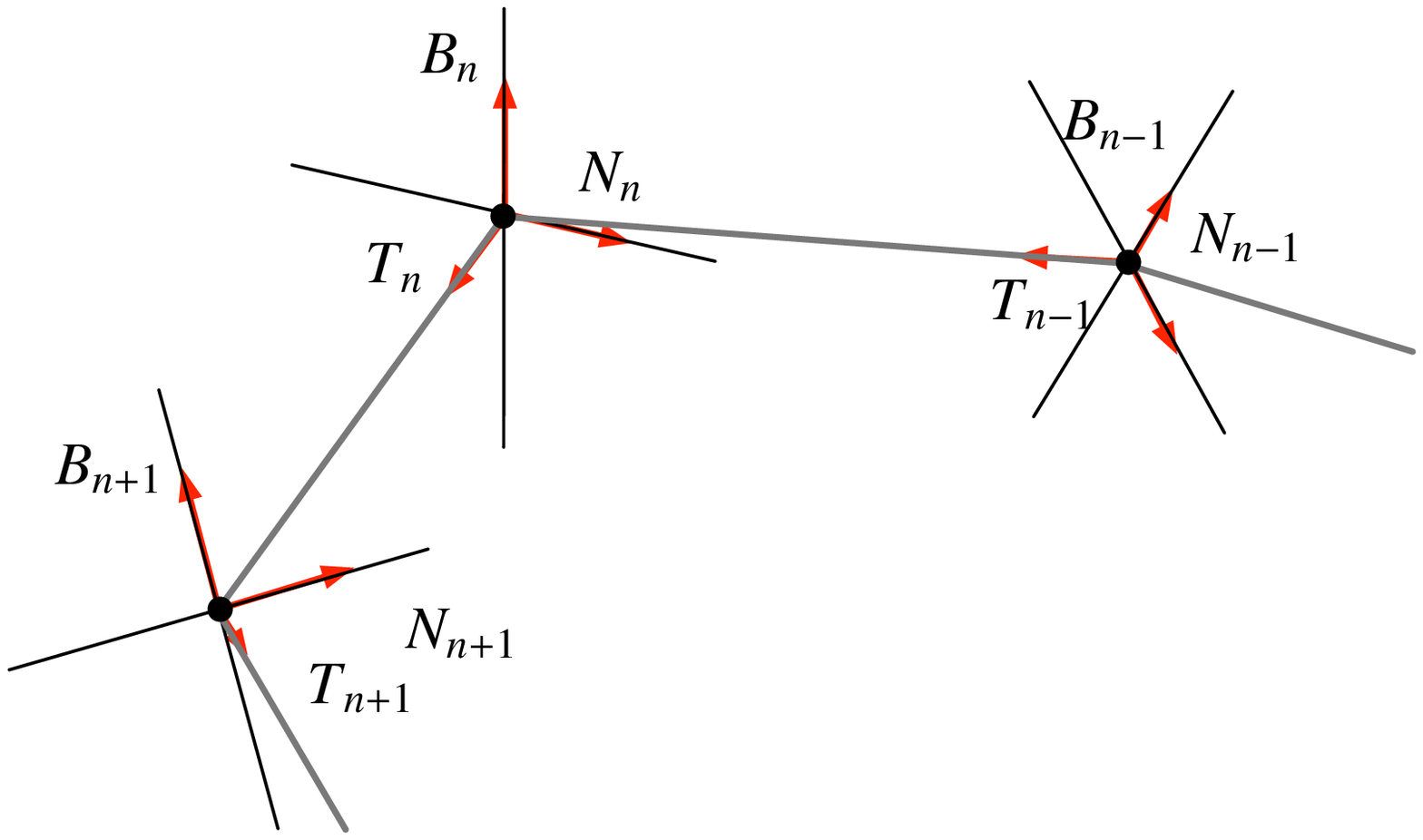}
  \includegraphics[scale=0.25]{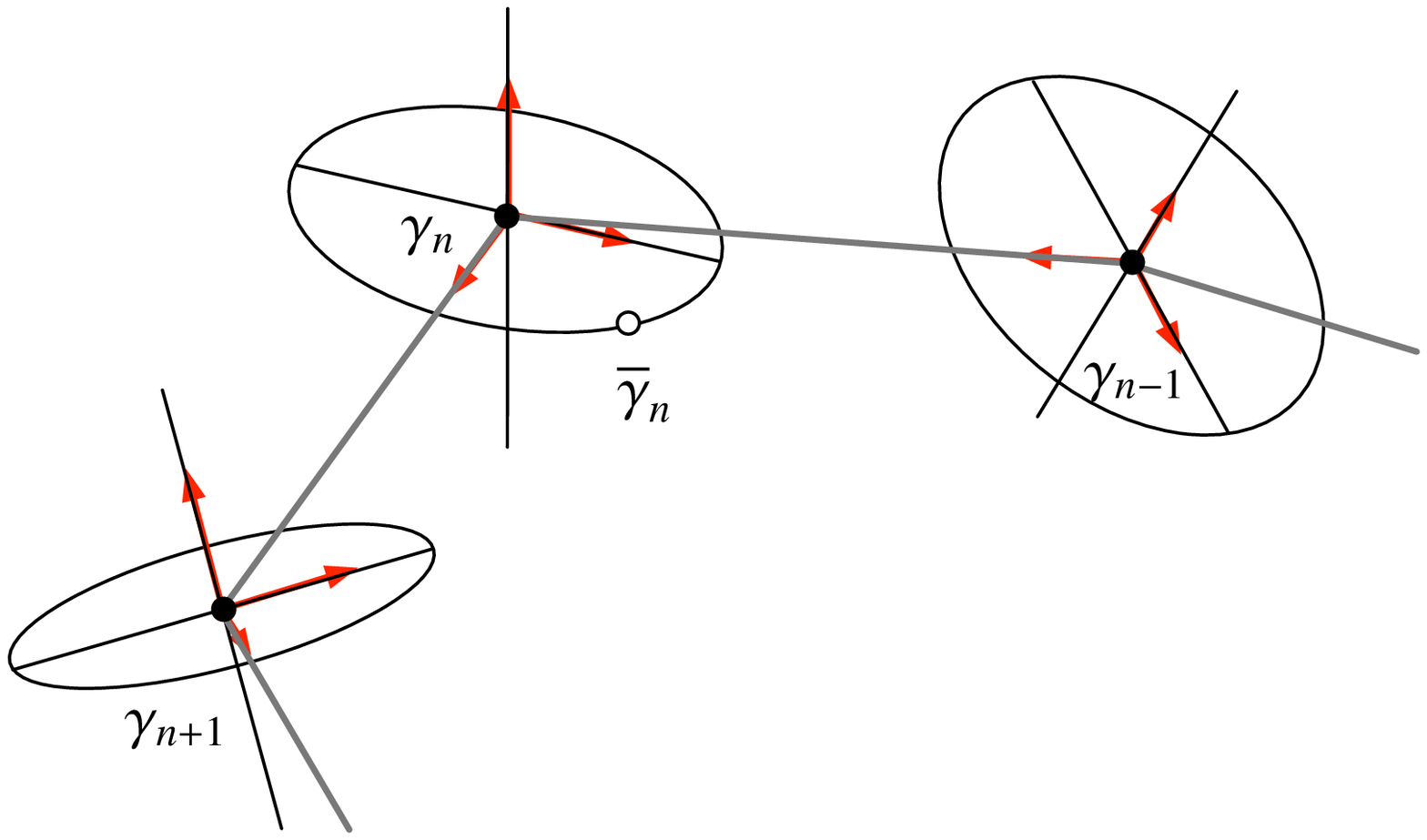}
  \includegraphics[scale=0.25]{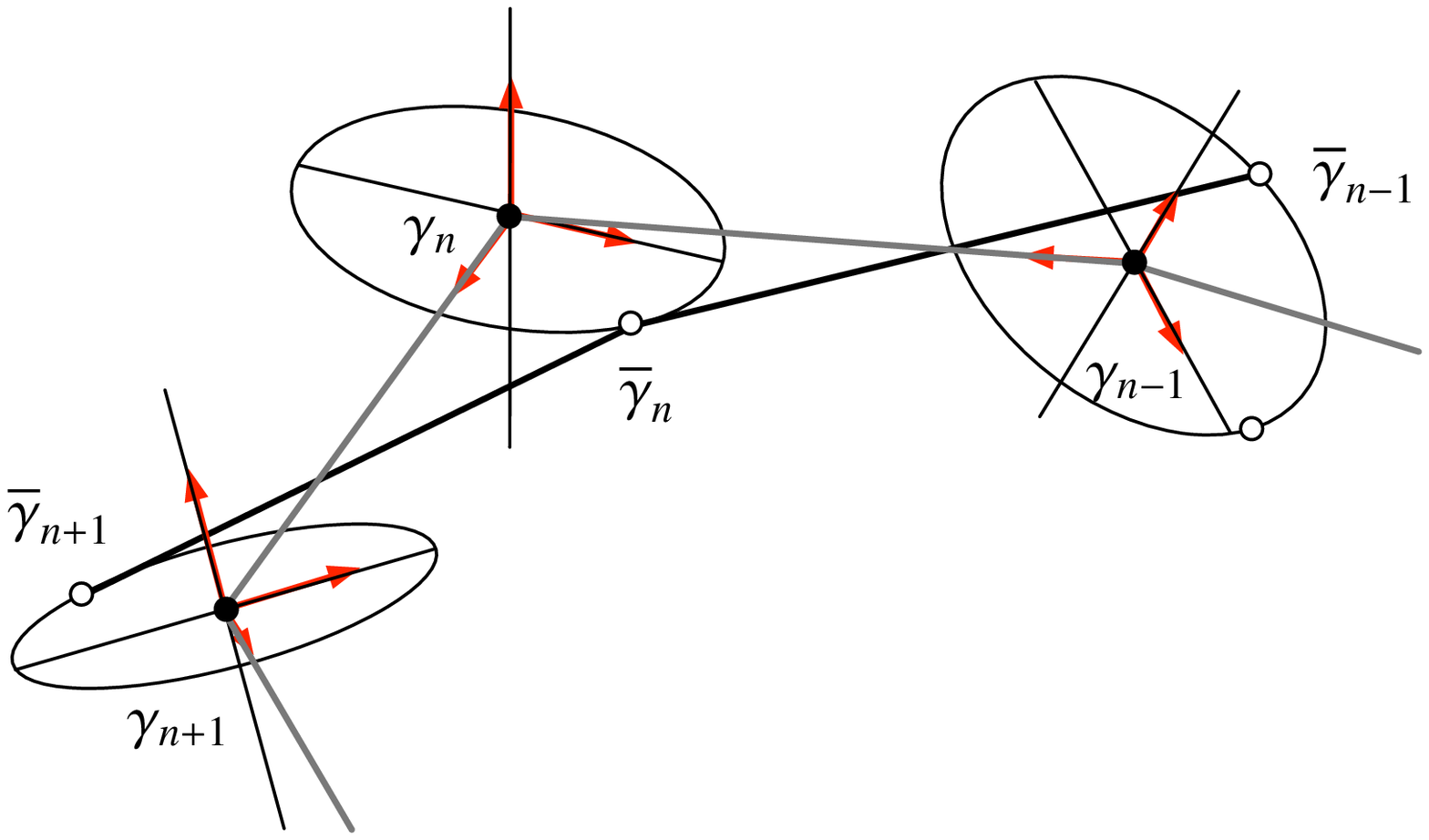} \caption{Deformation of a curve. Left: a
  discrete curve and its Frenet frame, middle: step (1) and (2), right:
  step(3)}\label{fig:deformation}
\end{center}
 \end{figure}
The condition that $\sin(w_{n+1}+\kappa_{n+1}-w_{n-1})$ is a constant for all $n$ is equivalent to
that $\overline{B}_n$ lies on the same side with respect to the plane spanned by
$\overline{\gamma}_n-\gamma_n$ and $B_n$ for all $n$ (see Fig.\ref{fig:torsion_const}), and this is
the necessary and sufficient condition that the deformation is torsion-preserving. We will discuss
this point in Section \ref{section:torsion_preserving}.
\begin{figure}[ht]
\begin{center}
  \includegraphics[scale=0.4]{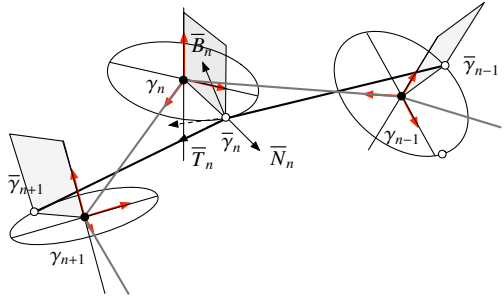}
\caption{The necessary and sufficient condition for the preservation of the torsion.
$\overline{B}$ should lie on the same side with respect to the plane spanned by $\overline{\gamma}-\gamma$ and $B$ (gray planes).}\label{fig:torsion_const}
\end{center}
\end{figure}
%
Repeating the construction in Proposition \ref{prop:dmkdvflow} yields the sequence of the 
discrete space curves with constant torsion $\gamma^0 = \gamma,\, \gamma^1 = \overline{\gamma},\, \gamma^2
= \overline{\gamma^1}, \ldots, \gamma^m=\overline{\gamma^{m-1}},\ldots$. Correspondingly, we write
the data $\kappa$, $T$, $N$, $B$ associated with the discrete curve as $\kappa^m$, $T^m$, $N^m$,
$B^m$, respectively. We also write the data of deformation $\delta$, $b$, $\mu$ as $\delta_m$,
$b_m$, $\mu_m$, respectively.  The following theorem is derived immediately by applying Proposition
\ref{prop:dmkdvflow} successively.
\begin{thm}\label{thm:dmkdvflow}\rm
Let $\gamma^0$ be a discrete space curve with the constant torsion $\lambda$.
We define the sequence of discrete space curves $\gamma^m$ by
\begin{gather}
\gamma^{m + 1}_n
=\label{def:isoperi}
\gamma^m_n
+
\delta_m
\left(\cos w^m_n T^m_n + \sin w^m_n N^m_n\right),\\
\delta_m
=\label{def:isoperi-c}
\dfrac{b_m}{1 + \frac{b_m^2\lambda^2}{4}}>0\, ,\\
w^m_{n + 1} 
=\label{def:isoperi-w}
-  \kappa^m_{n + 1} 
+ 2\arctan\dfrac{b_m + a_n}{b_m - a_n}\tan \dfrac{w^m_n}{2},
\end{gather}
where $a_n$ is determined from the data $\epsilon$, $\nu$ of the initial curve $\gamma^0$ by
\eqref{eqn:a_and_epsilon_nu}. We require that the sequences $b_m>0$ and $w_0^m$ should be
chosen so that the sign of $\sin(w_{n+1}^m+\kappa_{n+1}^m-w_{n-1}^m)$ is constant for all $n$.
Then we have the following:
\begin{enumerate}
\item (Isoperimetricity and preservation of the torsion) For all $m$, 
$\gamma^m$ satisfies
\begin{equation}
\left|\gamma^m_{n + 1} - \gamma^m_n\right|
=\epsilon_n,
\end{equation}
and the torsion of $\gamma^m$ is $\lambda$.
Namely, \eqref{def:isoperi} gives a isoperimetric and torsion-preserving deformation.
\item (Deformation of the Frenet frame)
The Frenet frame $\varPhi^m_n = \left[T^m_n, N^m_n, B^m_n\right]$ satisfies
\begin{equation}
\varPhi^m_{n + 1} = \varPhi^m_n L^m_n,\quad
\varPhi^{m + 1}_n = \varPhi^m_n M^m_n. \label{Phi}
\end{equation}
Here,
\begin{equation}
L^m_n
=\label{L}
R_1 \left(- \nu_{n + 1}\right)
R_3 \left(\kappa^m_{n + 1}\right),\quad
\nu_{n + 1} = 2 \arctan \dfrac{a_n \lambda}{2},
\end{equation}
and $M_n^m$ is given by either of the following for each $m$: in case of
$\sin(w_{n+1}^m+\kappa_{n+1}^m-w_{n-1}^m)>0$,
\begin{equation}
M^m_n
=\label{M}
R_3 \left(w^m_n\right)
R_1 \left(\mu_m\right)
R_3 \left(w^m_{n + 1} + \kappa^m_{n + 1}\right),\quad 
\mu_m=- 2 \arctan \dfrac{b_m\lambda}{2},
\end{equation}
\noindent or in case of $\sin(w_{n+1}^m+\kappa_{n+1}^m-w_{n-1}^m)<0$,
\begin{equation}
M^m_n
=\label{M_sG}
R_3 \left(w^m_n\right)
R_1 \left(\mu_m\right)
R_3 \left(-w^m_{n + 1} - \kappa^m_{n + 1}\right),\quad 
\mu_m= 2 \arctan \dfrac{2}{b_m\lambda}.
\end{equation}
%
Choosing the matrix $M$ as \eqref{M}, then $L, M$ are the Lax pair of the discrete mKdV equation. Namely,
the compatibility condition $L_n^m M_{n+1}^m = M_n^m L_{n}^{m+1}$ for the system of difference equation 
with respect to $\varPhi^m_n$ \eqref{Phi} gives
\begin{equation}\label{compatibility_dmKdV}
 w_{n+1}^m-w_{n-1}^m = \kappa_{n}^{m+1}-\kappa_{n+1}^m.
\end{equation}
Then \eqref{def:isoperi-w} and \eqref{compatibility_dmKdV} yield the discrete mKdV equation 
\begin{equation}\label{dmkdv}
\dfrac{w^{m + 1}_{n + 1}}{2} - \dfrac{w^m_n}{2}
=
 \arctan \left(\dfrac{b_{m + 1} + a_n}{b_{m + 1} - a_n}\tan \dfrac{w^{m + 1}_n}{2}\right)
 - \arctan \left(\dfrac{b_m + a_{n + 1}}{b_m - a_{n + 1}}
\tan \dfrac{w^m_{n + 1}}{2}\right).
\end{equation}
Choosing the matrix $M$ as \eqref{M_sG}, then $L, M$ are the Lax pair of the discrete sine-Gordon equation.
Namely, the compatibility condition of \eqref{Phi} gives
\begin{equation}\label{compatibility_sG}
 w_{n+1}^m-w_{n-1}^m = -\kappa_{n}^{m+1}-\kappa_{n+1}^m,
\end{equation}
and then \eqref{def:isoperi-w} and \eqref{compatibility_sG} yields the discrete sine-Gordon equation
\begin{equation}\label{dsG}
\dfrac{w^{m + 1}_{n + 1}}{2} + \dfrac{w^m_n}{2}
=
 \arctan \left(\dfrac{b_{m + 1} + a_n}{b_{m + 1} - a_n}
\tan \dfrac{w^{m + 1}_n}{2}\right)
 + \arctan \left(\dfrac{b_m + a_{n + 1}}{b_m - a_{n + 1}}
\tan \dfrac{w^m_{n + 1}}{2}\right).
\end{equation}
\end{enumerate}
\end{thm}
\begin{rem}\rm\hfill
\begin{enumerate}
\item The deformation of the form \eqref{def:isoperi}, \eqref{def:isoperi-c} has been obtained in
       \cite{Calini-Ivey} as the B\"acklund transformation of the space smooth curve with constant
       torsion.
\item The discrete sine-Gordon equation \eqref{dsG} can be rewritten as the well-known
       form\cite{Hirota:dsG}
\begin{equation}
 \sin\frac{\theta_{n+1}^{m+1}-\theta_n^{m+1}-\theta_{n+1}^m+\theta_n^m}{4} = \frac{a_n}{b_m}
\sin\frac{\theta_{n+1}^{m+1}+\theta_n^{m+1}+\theta_{n+1}^m+\theta_n^m}{4}
\end{equation}
by the dependent variable transformation
\begin{equation}\label{theta_sG}
 w_n^m = -\frac{\theta_n^{m+1}+\theta_{n+1}^m}{2}.
\end{equation}
\end{enumerate}
\end{rem}
It is also possible for this case to reconstruct the discrete curves from the Frenet frame by 
using the Sym-Tafel formula. Lifting the Frenet frame $\varPhi=[T,N,B]\in\mathrm{SO}(3)$
to the $\mathrm{SU}(2)$-valued function $\phi$ by using Proposition \ref{prop:SO(3)_SU(2)},
then $\phi$ satisfies
\begin{align}
&\phi_{n + 1}^m
=\label{akns.dmkdv.n}
\phi_n^m L_n^m,\quad
L_n^m
=
\pm
\left[\begin{array}{cc}
 e^{-\frac{\sqrt{- 1}}{2} \nu_{n+1}} & 0\\
0 & e^{\frac{\sqrt{- 1}}{2} \nu_{n+1}}
\end{array}\right]
R\left(\frac{\kappa_{n+1}^m}{2}\right),\quad
\nu_{n+1}=2\arctan\frac{a_{n}\lambda}{2},
\end{align}
and 
\begin{equation}\label{akns.dmkdv.m}
\begin{split}
\phi_n^{m+1}
=
\phi_n^m M_n^m,\quad
&M_n^m
=\pm
R\left(\frac{w_n^m}{2}\right)
\left[\begin{array}{cc}
 e^{\frac{\sqrt{- 1}}{2} \mu_{m}} & 0\\
0 & e^{-\frac{\sqrt{- 1}}{2} \mu_{m}}
\end{array}\right]
R\left(\frac{w_{n+1}^m+\kappa_{n+1}^m}{2}\right),\\
&\mu_{m}=-2\arctan\frac{b_{m}\lambda}{2} ,
\end{split}
\end{equation}
or
\begin{equation}\label{akns.dsG.m}
\begin{split}
\phi_n^{m+1}
=
\phi_n^m M_n^m,\quad
&M_n^m
=\pm
R\left(\frac{w_n^m}{2}\right)
\left[\begin{array}{cc}
 e^{\frac{\sqrt{- 1}}{2} \mu_{m}} & 0\\
0 & e^{-\frac{\sqrt{- 1}}{2} \mu_{m}}
\end{array}\right]
R\left(-\frac{w_{n+1}^m+\kappa_{n+1}^m}{2}\right),\\
&\mu_{m}=2\arctan\frac{2}{b_{m}\lambda},
 \end{split}
\end{equation}
where
\begin{equation}
 R(\theta)=\left[\begin{array}{cc}\cos\theta & -\sin\theta\\ \sin\theta & \cos\theta\end{array}\right].
\end{equation}
\begin{thm}\label{thm:dmKdV_dsG}\rm
For a constant $\lambda$ and sequences $a_n$, $b_m$, $w_n^0$, $w_0^m$, we define the functions $w$
and $\kappa$ by \eqref{dmkdv} and \eqref{def:isoperi-w} respectively, and let $\phi$ be a solution
to the system of difference equations \eqref{akns.dmkdv.n}, \eqref{akns.dmkdv.m}.
Then let $f$ be the isomorphism defined by \eqref{su2_so3:f} and put
\begin{equation}\label{eqn:Sym_dmKdV}
\gamma^m_n
= f^{- 1} \left(S^m_n\right),\quad
S^m_n
= - \left(\dfrac{\partial}{\partial \lambda} \phi^m_n\right)
\left(\phi^m_n\right)^{- 1}.
\end{equation}
Then, for each $m$, $\gamma$ is a discrete space curve with the following properties: (1) the
distances between the contiguous vertices are given by
\begin{equation}
\left|\gamma_{n+1}^m-\gamma_n^m\right|=\epsilon_n
=
\dfrac{a_n}{1+ \frac{a_n^2\lambda^2}{4}},\quad
\left|\gamma_{n}^{m+1}-\gamma_n^m\right|=\delta_m
=
\dfrac{b_m}{1+ \frac{b_m^2\lambda^2}{4}},
\end{equation}
(2) the angles between the contiguous tangent vectors are $\kappa_n^m$, (3) the angles between the
contiguous binormal vectors are $\nu_n^m$, (4) the angles between $\gamma_n^{m+1}-\gamma_n^m$ and
$\gamma_{n+1}^m-\gamma_n^m$ are $w_n^m$.  Namely, $\gamma$ is a torsion-preserving isoperimetric and
equidistant deformation of the discrete space curves with constant torsion described by the discrete
mKdV equation.  Similarly, if we determine the function $w$ and $\kappa$ by \eqref{dsG} and
\eqref{def:isoperi-w} respectively, and let $\phi$ be a solution to the system of difference
equations \eqref{akns.dmkdv.n}, \eqref{akns.dsG.m}, then \eqref{eqn:Sym_dmKdV} gives a
torsion-preserving isoperimetric and equidistant deformation of the discrete space curves with
constant torsion described by the discrete sine-Gordon equation.
\end{thm}
In Theorem \ref{thm:dmkdvflow} (2), it is the necessary and sufficient condition for the deformation
of the discrete curve $\gamma^m\mapsto \gamma^{m+1}$ being torsion-preserving that the sign of
$\sin(w_{n+1}^m+\kappa_{n+1}^m-w_{n-1}^m)$ is constant as a function of $n$ at each $m$.  On the
other hand, given an initial curve $\gamma_n^0$ (accordingly $a_n$, $\kappa_n^0$ and $\lambda$), the
deformation of the curve is uniquely determined if we specify $w_0^m$ and $b_m$. Therefore, the
necessary and sufficient condition for the preservation of the torsion may be described as the
condition of $w_0^m$ and $b_m$ for given $a_n$, $\kappa_n^0$ and $\lambda$. However, it is
practically difficult to write down this condition.  Instead, we can show the following proposition
as a sufficient condition.
\begin{prop}\label{prop:torsion_preservation}\rm\hfill
\begin{enumerate}
 \item In the discrete deformation of the discrete curves in Theorem \ref{thm:dmkdvflow},
if we choose $b_m>0$ to satisfy either of the following conditions at each $m$,
then the deformation is torsion-preserving.
\begin{equation}
 \begin{split}
\mbox{(i)}&\quad  b_m > \max\left\{
\frac{a_{\rm max}}{\tan\frac{\kappa_{\rm min}^m}{4}},
\frac{\Delta a}{2\cos\frac{\kappa_{\max}^m}{2}}
\left(1+\sqrt{1+\frac{4a_{\min}a_{\max}}{(\Delta a)^2}\cos^2\frac{\kappa_{\rm max}^m}{2}}\right)
\right\}\\
&\quad\mbox{or}\\
\mbox{(ii)}&\quad b_m < \min\left\{
a_{\rm min}\tan\frac{\kappa_{\rm min}^m}{4},
\frac{\Delta a}{2\cos\frac{\kappa_{\max}^m}{2}}
\left(-1+\sqrt{1+\frac{4a_{\min}a_{\max}}{(\Delta a)^2}\cos^2\frac{\kappa_{\rm max}^m}{2}}\right)
\right\},
 \end{split}
\end{equation}
where
\begin{equation}
 \kappa_{\mathrm{min}}^m=\min_{n}\kappa_n^m,\ 
 \kappa_{\mathrm{max}}^m=\max_{n}\kappa_n^m,\ 
a_{\rm max}=\max_{n} a_n,\quad a_{\rm min}=\min_{n} a_n,\ 
\Delta a = a_{\rm max}-a_{\rm min} .
\end{equation}
 \item In the case of (i), the Frenet frame $\Phi$ is deformed according to \eqref{Phi}, \eqref{L},
       \eqref{M_sG}, and in the case of (ii) it is deformed according to \eqref{Phi}, \eqref{L},
       \eqref{M}. Namely, (i) and (ii) correspond to the deformation by the discrete sine-Gordon
       equation \eqref{dsG} and discrete mKdV equation \eqref{dmkdv} respectively.
\end{enumerate}
\end{prop}
\section{Proof of Main Results}
\subsection{Proof of Proposition \ref{prop:dmkdvflow}}\label{subsection_main1}
\subsubsection{Isoperimetricity}
We show that $|\overline{\gamma}_{n+1}-\overline{\gamma}_n|=\epsilon_n$ holds if we give the
deformation of the curves by \eqref{def:isoperi0} and $|\gamma_{n+1}-\gamma_n|=\epsilon_n$. Since we
have from direct computation that
\begin{align*}
|\overline{\gamma}_{n+1}-\overline{\gamma}_n|^2 &
= \epsilon_n^2 + 2\delta^2\left[1-\cos(w_{n+1} + \kappa_{n+1})\cos w_n 
- \cos\nu_{n+1}\sin(w_{n+1} + \kappa_{n+1})\sin w_n\right]\\
 &+ 2\epsilon_n\delta\left[ \cos(w_{n+1}+\kappa_{n+1})-\cos w_n\right],
\end{align*}
it suffices to show
\begin{equation*}
\begin{split}
\delta&\left[1 - \cos(w_{n+1} + \kappa_{n+1})\cos w_n
- \cos\nu_{n+1}\sin(w_{n+1} + \kappa_{n+1})\sin w_n\right]\\
&+ \epsilon_n\left[ \cos(w_{n+1} + \kappa_{n+1})-\cos w_n\right]=0 .
\end{split}
\end{equation*}
Noticing that 
\begin{displaymath}
\epsilon_n = \frac{a_n}{1+\frac{a_n^2\lambda^2}{4}},\quad
\delta = \frac{b}{1+\frac{b^2\lambda^2}{4}},\quad
\lambda=\frac{2}{a_n}\tan\frac{\nu_{n+1}}{2},
\end{displaymath}
the above equation can be rewritten as
\begin{equation}\label{discrete_isoperi}
\begin{split}
&\left(
\frac{\lambda^2a_nb}{4}
\sin\frac{w_{n+1} + \kappa_{n+1} + w_n}{2}
-
\sin\frac{w_{n+1} + \kappa_{n+1} - w_n}{2}
\right)\\
&\times
 \left(a_n \sin\frac{w_{n+1} + \kappa_{n+1} + w_n}{2} 
- b \sin\frac{w_{n+1} + \kappa_{n+1} - w_n}{2}\right)=0,
\end{split}
\end{equation}
which holds since the second factor becomes $0$ by \eqref{def:isoperi-w0}.
\begin{rem}\rm
The deformation such that the first factor of \eqref{discrete_isoperi} becomes $0$ is also
consistent with \eqref{def:isoperi0} and \eqref{def:isoperi-c0}.  See Remark
\ref{rem:alternate_deformation} in the end of this section.
\end{rem}
\subsubsection{Preservation of the Torsion}\label{section:torsion_preserving}
We next show that $\nu$ is invariant by the deformation. Since $\nu_{n+1}\in[-\pi,\pi)$, we have to
show that $\cos\overline{\nu}_{n+1}=\cos\nu_{n+1}$ and $\sin\overline{\nu}_{n+1}=\sin\nu_{n+1}$.
Noticing that $\cos\nu_{n+1}=\langle B_{n},B_{n+1}\rangle$ and $\sin\nu_{n+1}=\langle
N_{n},B_{n+1}\rangle$, we start from the computation of $\overline{T}_n$. From \eqref{def:isoperi0},
we introduce the displacement vector $D_n$ by
\begin{equation}\label{S}
D_n=\frac{\overline{\gamma}_{n}-\gamma_n}{\delta} = \varPhi_n \left[\begin{array}{c}\cos w_n \\ \sin w_n \\ 0  \end{array}\right],
\end{equation}
then it follows by definition that 
\begin{equation}
 \delta\, D_n + \epsilon_n\, \overline{T}_n = \epsilon_n\, T_n + \delta\, D_{n+1}.
\end{equation}
We have by noticing $\lambda=\frac{2}{a_n}\tan\frac{\nu_{n+1}}{2}$
\begin{align}
&\overline{T}_n
=\frac{\delta}{\epsilon_n}D_{n+1}
- \frac{\delta}{\epsilon_n}D_n
+T_n
=\frac{\delta}{\epsilon_n}\varPhi_{n+1}
\left[\begin{array}{c}\cos w_{n+1}\\\sin w_{n+1}\\0\end{array}\right]
- \frac{\delta}{\epsilon_n}\varPhi_{n}
\left[\begin{array}{c}\cos w_{n}\\\sin w_{n}\\0\end{array}\right]
+ \varPhi_{n}\left[\begin{array}{c}1\\0\\0\end{array}\right] \nonumber\\
&= \varPhi_n^m\left\{
\frac{\delta}{\epsilon_n}\left(L_n \left[\begin{array}{c}\cos w_{n+1}\\\sin w_{n+1}\\0\end{array}\right]
- \left[\begin{array}{c}\cos w_{n}\\\sin w_{n}\\0\end{array}\right]
\right) 
+ \left[\begin{array}{c}1\\0\\0\end{array}\right]
\right\} \nonumber\\
&=\frac{1}{1+\frac{b^2\lambda^2}{4}}\varPhi_n
\left[\begin{array}{c}
\medskip
\cos 2U_n
+ \dfrac{b^2\lambda^2}{4}\cos 2V_n\\
\medskip
\sin 2U_n
- \dfrac{b^2\lambda^2}{4}\sin 2V_n\\
-b\lambda\sin (U_n+V_n)
\end{array}\right]. \label{T_m+1}
\end{align}
Here, we put 
\begin{equation}\label{UV}
 U_n = \frac{w_{n+1} + \kappa_{n+1} + w_n}{2},\quad
 V_n = \frac{w_{n+1} + \kappa_{n+1} - w_n}{2}
\end{equation}
and used 
\begin{equation}
 a_n\sin U_n = b\sin V_n,
\end{equation}
which follows from \eqref{def:isoperi-w0}.
Moreover, from
\begin{align*}
 \overline{T}_{n-1} &= \frac{1}{1+\frac{b^2\lambda^2}{4}}\varPhi_{n-1}
\left[\begin{array}{c}
\medskip
\cos 2U_{n-1}
+ \dfrac{b^2\lambda^2}{4}\cos 2V_{n-1}\\
\medskip
\sin 2U_{n-1}
- \dfrac{b^2\lambda^2}{4}\sin 2V_{n-1}\\
-b\lambda\sin (U_{n-1}+V_{n-1})
\end{array}\right]\\
&= \frac{1}{1+\frac{b^2\lambda^2}{4}}\varPhi_{n}
\left(L_{n-1}\right)^{-1}
\left[\begin{array}{c}
\medskip
\cos 2U_{n-1}
+ \dfrac{b^2\lambda^2}{4}\cos 2V_{n-1}\\
\medskip
\sin 2U_{n-1}
- \dfrac{b^2\lambda^2}{4}\sin 2V_{n-1}\\
-b\lambda\sin (U_{n-1}+V_{n-1})
\end{array}\right],
\end{align*}
a tedious but straightforward calculation  by using \eqref{def:isoperi-w0} yields
\begin{align}
 \overline{T}_{n-1}\times \overline{T}_n &=
\frac{\sin(w_{n+1} + \kappa_{n+1} - w_{n-1})}{1+\frac{b^2\lambda^2}{4}}
\varPhi_n \left[\begin{array}{c}
-b\lambda\sin w_n \\
b\lambda\cos w_n\\ 
1 - \frac{b^2\lambda^2}{4}   \end{array}\right],
\end{align}
from which we obtain
\begin{equation}
 \overline{B}_{n}  =
\frac{\overline{T}_{n-1}\times \overline{T}_n}{\left|  \overline{T}_{n-1}\times \overline{T}_n \right|}
=\frac{\sgn(\sin(w_{n+1} + \kappa_{n+1} - w_{n-1}))}{1+\frac{b^2\lambda^2}{4}}\varPhi_n 
\left[\begin{array}{c}
-b\lambda\sin w_n \\
b\lambda\cos w_n\\ 
1 - \frac{b^2\lambda^2}{4}   \end{array}\right],
\label{B_m+1} 
\end{equation}
\begin{align}
  \overline{B}_{n+1}& =
\frac{\sgn(\sin(w_{n+2}+\kappa_{n+2}-w_{n}))}{1+\frac{b^2\lambda^2}{4}}\varPhi_{n} 
L_n\left[\begin{array}{c}
-b\lambda\sin w_{n+1} \\
b\lambda\cos w_{n+1}\\ 
1 - \frac{b^2\lambda^2}{4}
\end{array}\right] \nonumber\\
&= \frac{\sgn(\sin(w_{n+2}+\kappa_{n+2}-w_{n}))}{1+\frac{b^2\lambda^2}{4}}\varPhi_{n} 
\left[\begin{array}{c}
\medskip
-b\lambda\sin(\kappa_{n+1}+w_{n+1}) \\
\medskip
\frac{\lambda\left\{
b(1-\frac{a_n^2\lambda^2}{4}) \cos(\kappa_{n+1}+w_{n+1})
+ a_n (1-\frac{b^2\lambda^2}{4})\right\}}{1 + \frac{a_n^2\lambda^2}{4} }\\ 
-\frac{a_nb\lambda^2\cos(\kappa_{n+1}+w_{n+1}) 
- (1-\frac{a_n^2\lambda^2}{4})(1-\frac{b^2\lambda^2}{4})}
{1 + \frac{a_n^2\lambda^2}{4} }
\end{array}\right] .
\end{align}
Moreover, from $\overline{N}_n=\overline{B}_n\times \overline{T}_n$ we have
\begin{equation}\label{N_m+1}
 \overline{N}_n
=\frac{\sgn(\sin(w_{n+1} + \kappa_{n+1} - w_{n-1}))}{1+\frac{b^2\lambda^2}{4}}\varPhi_n
\left[\begin{array}{c}
\medskip
 -\sin 2U_n - \dfrac{b^2\lambda^2}{4}\sin 2V_n\\
\medskip
  \cos 2U_n - \dfrac{b^2\lambda^2}{4}\cos 2V_n\\
-b\lambda\cos(U_n+V_n)
      \end{array}\right].
\end{equation}
Therefore, if $\sgn(\sin(w_{n+2}+\kappa_{n+2}-w_{n}))=\sgn(\sin(w_{n+1}+\kappa_{n+1}-w_{n-1}))$ holds, then we have
by using \eqref{def:isoperi-w0} that 
\begin{align}
&\cos\overline{\nu}_{n+1} =
\langle \overline{B}_n, \overline{B}_{n+1} \rangle
=\frac{1-\frac{a_n^2\lambda^2}{4}}{1+\frac{a_n^2\lambda^2}{4}}
=\frac{1-\tan^2\frac{\nu_{n+1}}{2}}{1+\tan^2\frac{\nu_{n+1}}{2}}
=\cos\nu_{n+1},\\
&\sin\overline{\nu}_{n+1} =
\langle \overline{B}_{n+1}, \overline{N}_{n} \rangle
=\frac{a_n\lambda}{1+\frac{a_n^2\lambda^2}{4}}
=\frac{2\tan\frac{\nu_{n+1}}{2}}{1+\tan^2\frac{\nu_{n+1}}{2}}
=\sin\nu_{n+1},
\end{align}
which implies $\overline{\nu}_{n+1}=\nu_{n+1}$.
\qed
\begin{rem}\rm
The above discussion shows that the condition that $\sgn(\sin(w_{n+1}+\kappa_{n+1}-w_{n-1}))$ is
constant with respect to $n$ is the necessary and sufficient condition for the deformation being
torsion-preserving. Since we see from \eqref{S} and \eqref{B_m+1} that
$\sgn(\sin(w_{n+1}+\kappa_{n+1}-w_{n-1}))=\sgn(\langle D_n\times B_n, \overline{B}_n\rangle)$, the
geometric meaning of this condition is that $\overline{B}_n$ lie on the same side with respect to the plane
spanned by $D_n$ and $B_n$ for all $n$.
\end{rem}
\subsubsection{Deformation of the Frenet Frame}\label{subsection:time_evolution}
For $\varPhi_n=[T_n,N_n,B_n]\in{\rm SO}(3)$, we construct $M_n\in{\rm SO}(3)$ satisfying
\begin{equation}
\varPhi_{n+1} = \varPhi_n L_n,\quad \overline{\varPhi}_{n} = \varPhi_n M_n.
\end{equation}
In the calculations in the previous section, since the column vectors of 
$M_n$ are given by the right-hand side of \eqref{T_m+1}, \eqref{N_m+1} and \eqref{B_m+1}, respectively, 
we have for $\sin(w_{n+1}+\kappa_{n+1}-w_{n-1})>0$
\begin{align}
M_n &= \frac{1}{1+\frac{b^2\lambda^2}{4}}
\begin{small}
\left[
\begin{array}{ccc}
\cos 2U_n + \dfrac{b^2\lambda^2}{4}\cos 2V_n 
& - \sin 2U_n - \dfrac{b^2\lambda^2}{4}\sin 2V_n
& - b\lambda\sin w_n \\
\sin 2U_n - \dfrac{b^2\lambda^2}{4}\sin 2V_n 
&   \cos 2U_n - \dfrac{b^2\lambda^2}{4}\cos 2V_n
& b\lambda\cos w_n\\
-b\lambda\sin (U_n+V_n)  
& - b\lambda\cos(U_n+V_n)
&  1 - \frac{b^2\lambda^2}{4}
\end{array}
\right]
\end{small}
\\
&=
R_3(w_n)R_1(\mu)R_3(w_{n+1}+\kappa_{n+1}),\quad
\mu = -2\arctan\frac{b\lambda}{2}.
\end{align}
Similarly, we have for $\sin(w_{n+1}+\kappa_{n+1}-w_{n-1})<0$
\begin{align}
M_n &= \frac{1}{1+\frac{b^2\lambda^2}{4}}
\begin{small}
\left[
\begin{array}{ccc}
\cos 2U_n + \dfrac{b^2\lambda^2}{4}\cos 2V_n 
& \sin 2U_n + \dfrac{b^2\lambda^2}{4}\sin 2V_n
& b\lambda\sin w_n \\
\sin 2U_n - \dfrac{b^2\lambda^2}{4}\sin 2V_n 
&  - \cos 2U_n + \dfrac{b^2\lambda^2}{4}\cos 2V_n
& - b\lambda\cos w_n\\
-b\lambda\sin (U_n+V_n)  
&  b\lambda\cos(U_n+V_n)
& -1 + \frac{b^2\lambda^2}{4}
\end{array}
\right] 
\end{small}\\
&=
R_3(w_n)R_1(\mu)R_3(-w_{n+1} - \kappa_{n+1}),\quad
\mu = 2\arctan\frac{2}{b\lambda}.
\end{align}
\qed
\subsection{Proof of Theorem \ref{thm:dmKdV_dsG}}
If the Frenet frame $\varPhi= \left[T, N, B\right]\in\mathrm{SO}(3)$ satisfies \eqref{Phi}, \eqref{L} and \eqref{M},
the corresponding $\mathrm{SU} (2)$-valued function $\phi$ satisfy \eqref{akns.dmkdv.n} and \eqref{akns.dmkdv.m}.
In particular, \eqref{akns.dmkdv.m} is rewritten as
\begin{equation}\label{akns.dmkdv.m_1}
\phi_n^{m+1}
=
\phi_n^m M_n^m,\ 
M_n^m
=
\dfrac{\pm 1}{\sqrt{1+\frac{b_m^2\lambda^2}{4}}}
\left[\begin{array}{cc}\medskip
\cos U_n^m - \frac{\sqrt{-1}b_m\lambda}{2}\cos V_n^m 
& -\sin U_n^m + \frac{\sqrt{-1}b_m\lambda}{2}\sin V_n^m \\
\sin U_n^m + \frac{\sqrt{-1}b_m\lambda}{2}\sin V_n^m 
& \cos U_n^m +\frac{\sqrt{-1}b_m\lambda}{2}\cos V_n^m 
\end{array}\right],
\end{equation}
where $U_n^m$ and $V_n^m$ are given by \eqref{UV}.
Putting
\begin{equation}
T^m_n = \phi_n^m e_1 \left(\phi_n^m\right)^{-1},\quad
N_n^m = \phi_n^m e_2 \left(\phi_n^m\right)^{-1},\quad
B_n^m = \phi_n^m e_3 \left(\phi_n^m\right)^{-1},
\end{equation}
$S_n^m=-(\phi_n^m)_\lambda(\phi_n^m)^{-1}$ satisfies
\begin{align}
S_n^{m+1} - S_n^m 
&= -\phi_n^m(M_n^m)_\lambda(M_n^m)^{-1}(\phi_n^m)^{-1} 
= -\phi_n^m\left\{\frac{-b_m}{1+\frac{b_m^2\lambda^2}{4}}
(\cos w_n^m e_1 + \sin w_n^m e_2)\right\}(\phi_n^m)^{-1} \nonumber\\
&= \delta_m\left(\cos w_n^m T_n^m + \sin w_n^m N_n^m\right),\label{discrete:S_deformation}
\end{align}
which coincides with the definition of the deformation \eqref{def:isoperi}. Similarly, if $\varPhi$
satisfies \eqref{Phi}, \eqref{L} and \eqref{M_sG}, then it follows that $\phi$ satisfies
\eqref{akns.dmkdv.n} and \eqref{akns.dsG.m}. Then \eqref{discrete:S_deformation} is derived in a similar manner.
\qed
%
\begin{rem}\label{rem:alternate_deformation}\rm
The deformation satisfying the condition
\begin{equation}
 \frac{\lambda^2a_nb}{4}
\sin\frac{w_{n+1} + \kappa_{n+1} + w_n}{2}
-
\sin\frac{w_{n+1} + \kappa_{n+1} - w_n}{2}=0
\end{equation}
or
\begin{equation}
w_{n + 1} 
=\label{def:isoperi-w-2}
-  \kappa_{n + 1} 
+ 2\arctan\dfrac{\hat{b} + a_n}{\hat{b} - a_n}\tan \dfrac{w_n}{2},\quad 
\hat{b} =\frac{4}{b\lambda^2}
\end{equation}
which eliminates the first factor of \eqref{discrete_isoperi} is consistent with
\eqref{def:isoperi0} and \eqref{def:isoperi-c0}. This deformation is obtained from the invariance of
$\delta$ with respect to the transformation $b\mapsto \hat{b}$. We do not discuss this deformation
further since all of the above results can be rewritten as those of this deformation by the
redefinition of the parameter $b\mapsto \hat{b}$.
\end{rem}
\subsection{Proof of Proposition \ref{prop:torsion_preservation}}
We fix $m$. Noticing \eqref{def:isoperi-w}, 
$\sin(w_{n+1}^m+\kappa_{n+1}^m-w_{n-1}^m)$ can be rewritten as 
\begin{equation}\label{app:sin}
\scalebox{0.9}{$\displaystyle
 \sin(w_{n+1}^m+\kappa_{n+1}^m-w_{n-1}^m)
=\sin\left(2\arctan\left(\frac{b_m+a_n}{b_m-a_n}\tan\frac{w_n^m}{2}\right)
-2\arctan\left(\frac{b_m-a_{n-1}}{b_m+a_{n-1}}\tan\frac{w_{n}^m+\kappa_n^m}{2}\right)\right).$}
\end{equation}
For simplicity, we put 
\begin{equation}
 \tan\frac{w_n^m}{2}=W,\quad \tan\frac{\kappa_n^m}{2}=K,\quad \quad c=\frac{b_m-a_n}{b_m+a_n},\quad
\underline{c}=\frac{b_m-a_{n-1}}{b_m+a_{n-1}}.
\end{equation}
Here since we have $a_n,b_m>0$, $\kappa_n^m\in(0,\pi)$ and $w_n^m\in[-\pi,\pi)$, it follows that
\begin{equation}\label{eqn:range_params}
 -1<c<1,\quad -1<\underline{c}<1,\quad K>0,\quad -\infty<W<\infty.
\end{equation}
Noticing
\begin{displaymath}
 \sin(2\arctan x - 2\arctan y) = \frac{2(1+xy)(x-y)}{(1+x^2)(1+y^2)},
\end{displaymath}
\eqref{app:sin} can be rewritten as
\begin{equation}\label{app:sin2}
 \sin(w_{n+1}^m+\kappa_{n+1}^m-w_{n-1}^m)=
\frac{-\left\{\underline{c}W^2+K(\underline{c}-c)W+c\right\}
\left\{KW^2+(c\underline{c}-1)W+Kc\underline{c}\right\}}{(1+\frac{W^2}{c^2})
\left[1+\left(\underline{c}\frac{W+K}{1-WK}\right)^2\right]
(1-KW)^2c^2}.
\end{equation}
Now we require that for arbitrary $W$ the sign of \eqref{app:sin2} is constant for all $n$.
To this end, since the denominator of the right-hand side does not affect the sign, 
it is sufficient that discriminants of the two factors of the numerator which are quadratic in $W$ are negative.
Therefore, we determine $b_m$ such that both of the inequalities
\begin{equation}\label{discriminant}
 \begin{split}
& K^2(c-\underline{c})^2-4c\underline{c}<0,\\
& (c\underline{c}-1)^2-4K^2c\underline{c} <0
 \end{split}
\end{equation}
are satisfied simultaneously. Solving the inequalities \eqref{discriminant}, we have
\begin{align}
&\left(\sqrt{K^2+1}-K\right)^2 < c\underline{c} < \left(\sqrt{K^2+1}+K\right)^2,\label{cc1}\\
& \left(\frac{\sqrt{K^2+1}-1}{K}\right)^2 < \frac{c}{\underline{c}} 
< \left(\frac{\sqrt{K^2+1}+1}{K}\right)^2. \label{cc2}
\end{align}
\eqref{cc1} and \eqref{cc2} imply that the contiguous $c_n$ with respect to $n$ must have the same sign, namely, 
for each $m$, $c_n$ must have the same sign for all $n$. So we put
\begin{equation}\label{amax}
a_{\rm max} = \max_n a_n ,\ a_{\rm min} = \min_n a_n,\ 
c_{\rm max} = \max_n c_n = \frac{b_m-a_{\rm min}}{b_m+a_{\rm min}},\ 
c_{\rm min} = \min_n c_n = \frac{b_m-a_{\rm max}}{b_m+a_{\rm max}},
\end{equation}
and we divide the discussion into two cases; (i) $c_{\rm min},c_{\rm max}>0$ ($b_m>a_{\rm max}$) and
(ii) $c_{\rm min},c_{\max}<0$ ($b_m < a_{\rm min}$).
%

\noindent(i) The case of $c_{\rm min},c_{\rm max}>0$ ($b_m>a_{\rm max}$). From \eqref{cc1},
\eqref{cc2}, it is sufficient to find the condition such that the following inequalities
\begin{align}
& \left(\sqrt{K^2+1}-K\right)^2 < c_{\rm min}^2\leq c_{\rm max}^2 < \left(\sqrt{K^2+1}+K\right)^2,
\label{cc3}\\
& \left(\frac{\sqrt{K^2+1}-1}{K}\right)^2 < \frac{c_{\rm min}}{c_{\rm max}}\leq
\frac{c_{\rm max}}{c_{\rm min}} < \left(\frac{\sqrt{K^2+1}+1}{K}\right)^2\label{cc4}
\end{align}
are satisfied simultaneously. Solving \eqref{cc3} and \eqref{cc4} in terms of $b_m$ by noticing
\eqref{eqn:range_params}, \eqref{amax} and
\begin{equation}
 \frac{\sqrt{K^2+1}-1}{K} = \tan\frac{\kappa_n^m}{4},\quad \frac{\sqrt{K^2+1}+1}{K} = 
\frac{1}{\tan\frac{\kappa_n^m}{4}},
\end{equation}
we have
\begin{equation}
b_m > \frac{a_{\rm max}}{\tan\frac{\kappa_n^m}{4}},\quad
 b_m > \frac{\Delta a}{2\cos^2\frac{\kappa_n^m}{2}}
\left(1+\sqrt{1+\frac{4a_{\rm min}a_{\max}}{(\Delta a)^2}\cos^2\frac{\kappa_n^m}{2}}\right),
\end{equation}
respectively. Note that the right-hand sides of the first and the second inequalities are monotonic
decreasing and monotonic increasing with respect to $\kappa_n^m$ respectively. Therefore it is
sufficient to choose $b_m$ as
\begin{equation}
 b_m > \max\left\{
\frac{a_{\rm max}}{\tan\frac{\kappa_{\rm min}^m}{4}},
\frac{\Delta a}{2\cos \frac{\kappa_{\rm max}^m}{2}}
\left(1+\sqrt{1+\frac{4a_{\rm min}a_{\max}}{(\Delta a)^2}\cos^2\frac{\kappa_{\rm max}^m}{2}}\right)
\right\},
\end{equation}
in order for those inequalities to hold for all $n$.

%

\noindent
(ii) The case of $c_{\rm min},c_{\rm max}<0$ ($b_m < a_{\rm max}$). From \eqref{cc1}, \eqref{cc2}
we have the following inequalities:
\begin{align}
& \left(\sqrt{K^2+1}-K\right)^2 < c_{\rm max}^2\leq c_{\rm min}^2 < \left(\sqrt{K^2+1}+K\right)^2,
\label{cc5}\\
& \left(\frac{\sqrt{K^2+1}-1}{K}\right)^2 < \frac{c_{\rm max}}{c_{\rm min}}\leq
\frac{c_{\rm min}}{c_{\rm max}} < \left(\frac{\sqrt{K^2+1}+1}{K}\right)^2.\label{cc6}
\end{align}
Solving these inequalities in a similar manner to the case of (i), we find that 
it is sufficient to choose $b_m$ as
\begin{equation}
 b_m < \min\left\{
a_{\rm min}\tan\frac{\kappa_{\rm min}^m}{4},
\frac{\Delta a}{2\cos^2\frac{\kappa_{\rm max}^m}{2}}
\left(-1+\sqrt{1+\frac{4a_{\rm min}a_{\max}}{(\Delta a)^2}\cos^2\frac{\kappa_{\rm max}^m}{2}}\right)
\right\}.
\end{equation}
This proves (1). As to (2), noticing that $K>0$ in the right hand side of \eqref{app:sin2}, the case
of $c_n>0$ (namely the case of (i)) corresponds to $\sin(w^m_{n+1}+\kappa_{n+1}^m-w_{n+1}^m)<0$, and
the case of $c_n<0$ (namely the case of (ii)) corresponds to
$\sin(w^m_{n+1}+\kappa_{n+1}^m-w_{n+1}^m)>0$ respectively. Moreover, according to Section 
\ref{subsection:time_evolution}, we see that the former and the latter cases correspond to
the deformation described by the discrete sine-Gordon equation and that described by the 
discrete mKdV equation respectively. \qed
\section{Relation to Discrete $K$-Surfaces}
The sequences of deformed discrete curves described in Theorem \ref{thm:dmkdvflow} form discrete
surfaces in $\R^3$. We show that it is nothing but the discrete $K$-surfaces.
\begin{defn}\label{def:discrete-K}\rm
If a map $\gamma:~\mathbb{Z}^2\rightarrow \mathbb{R}^3$, $(n,m)\mapsto \gamma_n^m$ satisfies the following conditions, we
call it the {\it discrete $K$-surface}.
\begin{enumerate}
 \item The five points $\gamma_n^m$, $\gamma_{n\pm 1}^m$, $\gamma_n^{m\pm 1}$ are coplanar.
If this condition is satisfied, we say that $\gamma$ forms the {\em discrete asymptotic net}.
 \item It holds that $|\gamma_{n+1}^m-\gamma_n^m|=|\gamma_{n+1}^{m+1}-\gamma_n^{m+1}|$ and
$|\gamma_{n}^{m+1}-\gamma_n^m|=|\gamma_{n+1}^{m+1}-\gamma_{n+1}^{m}|$.
\end{enumerate}
\end{defn}
\begin{prop}\rm
The sequence of deformed discrete curves $\gamma$ described in Theorem \ref{thm:dmkdvflow} form 
a discrete $K$-surface in $\R^3$.
\end{prop}
\begin{proof}
We first compute $\gamma_n^{m-1}$ in order to verify the condition (1) in Definition 
\ref{def:discrete-K}. From
\begin{displaymath}
 \frac{\gamma_n^{m+1}-\gamma_n^m}{\delta_m} =\cos w_n^m T_n^m + \sin w_n^m N_n^m
=\varPhi_n^m\left[\begin{array}{c}\cos w_n^m \\\sin w_n^m \\ 0\end{array}\right]
\end{displaymath}
and \eqref{Phi}, we have
\begin{displaymath}
 \frac{\gamma_n^{m}-\gamma_n^{m-1}}{\delta_{m-1}} =
\varPhi_n^{m-1}\left[\begin{array}{c}\cos w_n^{m-1} \\\sin w_n^{m-1} \\ 0 \end{array}\right] 
=
\varPhi_n^{m}(M_n^{m-1})^{-1}\left[\begin{array}{c}\cos w_n^{m-1} \\\sin w_n^{m-1} \\ 0 \end{array}\right] .
\end{displaymath}
Substituting \eqref{M} and \eqref{M_sG} as $M_n^m$, we get
\begin{equation*}
\gamma^{m - 1}_n
= \gamma^m_n - \delta_{m - 1}
\left(\cos \left(w^{m - 1}_{n + 1} + \kappa^{m - 1}_{n + 1}\right)
T^m_n
\mp \sin \left(w^{m - 1}_{n + 1} + \kappa^{m - 1}_{n + 1}\right)
N^m_n\right),
\end{equation*}
Here, the cases where $M_n^m$ is given by \eqref{M} and \eqref{M_sG} correspond to $-$ and $+$
respectively, in the right-hand side. This implies $\gamma_n^{m\pm 1}-\gamma_n^m\in{\rm
span}\{T_n^m,N_n^m\}$. Moreover, it is clear from the definition of $N_n^m$ that 
$\gamma_{n\pm 1}^{m}-\gamma_n^m\in{\rm span}\{T_n^m,N_n^m\}$. Therefore $\gamma$ forms the
discrete asymptotic net. Further, the condition (2) follows from 
$\left|\gamma^m_{n + 1} - \gamma^m_n\right| 
= \epsilon_n,\ \left|\gamma^{m + 1}_n - \gamma^m_n\right|= \delta_m$.
This proves that $\gamma$ forms a discrete $K$-surface.
\end{proof}
\begin{rem}\rm
Suppose that the Frenet frame $\varPhi$ satisfies \eqref{Phi}, \eqref{L} and \eqref{M_sG}.
Then putting 
\begin{equation*}
\varPsi^m_n
=
\varPhi^m_n
\left[\begin{array}{ccc}
0 & 0 & 1\\
1 & 0 & 0\\
0 & 1 & 0
\end{array}\right]
R_1 \left(\dfrac{- \theta^m_{n + 1} + \theta^m_n}{2}\right)
\end{equation*}
by using $\theta_n^m$ given in \eqref{theta_sG}, $\varPsi$ satisfies
\begin{align}
& \varPsi^m_{n + 1}
= \varPsi^m_n
R_1 \left(\dfrac{\theta^m_{n + 1} - \theta^m_n}{2}\right)
R_2 \left(- \nu_{n + 1}\right)
R_1 \left(\dfrac{\theta^m_{n + 1} - \theta^m_n}{2}\right),\\
& \varPsi^{m + 1}_n
= \varPsi^m_n
R_1 \left(- \dfrac{\theta^{m + 1}_n + \theta^m_n}{2}\right)
R_2 \left(\mu_m\right)
R_1 \left(\dfrac{\theta^{m + 1}_n + \theta^m_n}{2}\right).
\end{align}
Here, we have used 
\begin{equation}
 \kappa_n^m = \frac{\theta_{n+1}^m-\theta_{n-1}^m}{2}
\end{equation}
which follows from the compatibility condition \eqref{compatibility_sG} and \eqref{theta_sG}.
Lifting $\varPsi \in \SO (3)$ to $\psi \in \SU (2)$, 
$\psi$ satisfies
\begin{align}
& \psi^m_{n + 1}
= \psi^m_n
\left[\begin{array}{cc}\medskip
\cos \dfrac{\nu_{n + 1}}{2}
e^{\frac{\sqrt{- 1}}{2} (\theta^m_{n + 1} - \theta^m_n)} &
\sqrt{- 1} \sin \dfrac{\nu_{n + 1}}{2}\\
\sqrt{- 1} \sin \dfrac{\nu_{n + 1}}{2} &
\cos \dfrac{\nu_{n + 1}}{2}
e^{-\frac{\sqrt{- 1}}{2} (\theta^m_{n + 1} - \theta^m_n)}
\end{array}\right],\\[2mm]
&\psi^{m + 1}_n
= \psi^m_n
\left[\begin{array}{cc}\medskip
\cos \dfrac{\mu_m}{2} &
- \sqrt{- 1} \sin \dfrac{\mu_m}{2}
e^{-\frac{\sqrt{- 1}}{2} (\theta^{m+1}_{n} + \theta^m_n)}\\
- \sqrt{- 1} \sin \dfrac{\mu_m}{2}
e^{\frac{\sqrt{- 1}}{2} (\theta^{m+1}_{n} + \theta^m_n)}&
\cos \dfrac{\mu_m}{2}
\end{array}\right].
\end{align}
Noticing 
\begin{equation*}
\cos \dfrac{\nu_{n + 1}}{2}
=
\dfrac{1}{\sqrt{ 1+ \frac{a_n^2\lambda^2 }{4}}},\ 
\sin \dfrac{\nu_{n + 1}}{2}
=
\dfrac{\frac{a_n\lambda }{2}}{\sqrt{1 + \frac{a_n^2\lambda^2 }{4}}},\ 
\cos \dfrac{\mu_m}{2}
=
\dfrac{\frac{b_m\lambda}{2}}{\sqrt{1 + \frac{b_m^2\lambda^2 }{4}}},\ 
\sin \dfrac{\mu_m}{2}
=
\dfrac{1}{\sqrt{1  + \frac{b_m^2\lambda^2}{4}}},
\end{equation*}
we find that this $\SU (2)$ representation coincides with
the Lax pair of the discrete sine-Gordon equation constructed by
Bobenko and Pinkall\cite{Bobenko-Pinkall} , and that
$\varPsi$ gives the orthonormal frame of the discrete $K$-surface
presented in \cite{Bobenko-Pinkall}.
\end{rem}
\par\bigskip

\noindent\textbf{Acknowledgements.}\quad The authors would like to thank Professor Shimpei Kobayashi for
fruitful discussions. This work is partially supported by JSPS KAKENHI no. 22656026, 23340037,
24340029, 24540063 and 24540103.  \appendix
\section{Correspondence between Curves on Sphere and Space Curves with Constant Torsion}\label{section:app_A}
\begin{prop}\rm
Let $x$ be an arc-length parameter, and for $\lambda\in\mathbb{R}$ let 
$\Gamma(x)\in S^2(\frac{1}{|\lambda|})$ be a curve on $S^2(\frac{1}{|\lambda|})$.
Then
\begin{equation}\label{A:def_gamma}
 \gamma = \lambda\int \Gamma\times \Gamma'~dx\in\mathbb{R}^3
\end{equation}
is a space curve with the constant torsion $\lambda$. Conversely,
for a space curve $\gamma\in\mathbb{R}^3$ with the constant torsion $\lambda$,
let $B$ be its binormal vector. Then 
\begin{equation}\label{A:def_Gamma}
 \Gamma = \pm \frac{1}{\lambda}\,B
\end{equation}
is a curve on $S^2(\frac{1}{|\lambda|})$
\end{prop}
\begin{proof}
We first show that the torsion of $\gamma$ in \eqref{A:def_gamma} is $\lambda$.
By definition we have
\begin{equation}\label{A:Gamma1}
 \langle \Gamma,\Gamma\rangle = \frac{1}{\lambda^2},\quad
 \langle \Gamma',\Gamma'\rangle = 1.
\end{equation}
We define the Frenet frame $F$ of $\Gamma$ by
\begin{equation}
 F = \left[\lambda\,\Gamma\times \Gamma', \Gamma', -\lambda\,\Gamma\right]\in\SO(3).
\end{equation}
We compute $\Gamma''$ by noticing that $\langle \Gamma'',\Gamma\rangle=-1$, $\langle
\Gamma'',\Gamma'\rangle=0$ from \eqref{A:Gamma1}. Then we find that there exists a function $\kappa_G$
such that 
\begin{equation}
 \Gamma'' = \kappa_G (\lambda\,\Gamma\times\Gamma') + \lambda(-\lambda\Gamma)
\end{equation}
is satisfied. The tangent vector $T$ of $\gamma$ and $T'$ is computed as
\begin{align}
& T = \gamma'=\lambda\,\Gamma\times \Gamma'\\
& T' = \lambda\,\Gamma\times \Gamma'' = \lambda^2\kappa_G\Gamma\times(\Gamma\times \Gamma')
=\lambda^2\kappa_G(\langle \Gamma,\Gamma'\rangle \Gamma - \langle\Gamma,\Gamma\rangle \Gamma')
=-\kappa_G\Gamma',
\end{align}
respectively. Then the curvature of $\gamma$ is given by
\begin{equation}
 \kappa=|T'| = \left|\kappa_G\right|.
\end{equation}
Moreover, the normal vector and the binormal vector of $\gamma$ is computed as
\begin{align}
& N=\frac{T'}{\kappa} = \pm \Gamma'\\
& B=T\times N = \lambda(\Gamma\times \Gamma')\times \left(\pm \Gamma'\right)
=\mp \lambda\Gamma'\times(\Gamma\times \Gamma')=\mp \lambda\,\Gamma,
\end{align}
respectively. Therefore, the torsion $-\langle N,B'\rangle$ of $\gamma$ is given by
\begin{equation}
 -\langle N,B'\rangle = -\langle \pm \Gamma',\mp\lambda\,\Gamma'\rangle
=\lambda\langle \Gamma',\Gamma'\rangle
=\lambda,
\end{equation}
which proves the first half of the statement. The second half is obvious.
\end{proof}
\begin{prop}\rm
For $\lambda\in\mathbb{R}$, let $\Gamma_n\in S^2(\frac{1}{|\lambda|})$ be a discrete curve on
$S^2(\frac{1}{|\lambda|})$. Then, if the sign of $\langle \Gamma_{n-1}\times \Gamma_n,\Gamma_{n+1}\rangle$
is constant with respect to $n$, 
\begin{equation}\label{A:def_gamma_n}
 \gamma_n =\lambda\sum_{k}^{n-1} \Gamma_{k+1}\times \Gamma_{k}\in\mathbb{R}^3
\end{equation}
is a discrete space curve with the constant torsion $\lambda$.
Conversely, for a discrete space curve $\gamma_n\in\mathbb{R}^3$ with the constant torsion $\lambda$
let $B_n$ be its binormal vector. Then 
\begin{equation}\label{A:def_Gamma_n}
 \Gamma_n = \pm \frac{1}{\lambda}\,B_n
\end{equation}
is a discrete curve on $S^2(\frac{1}{|\lambda|})$.
\end{prop}
\begin{proof}
We only show the first half of the statement, since the second half is obvious.
We show that the torsion of $\gamma_n$ given in \eqref{A:def_gamma_n} is $\lambda$.
For $\Gamma_{n-1}, \Gamma_{n}\in S^2(\frac{1}{|\lambda|})$ let $\xi_{n}\in(0,\pi)$ be
the angle between $\Gamma_{n-1}$ and $\Gamma_{n}$ (see Fig. \ref{fig:sphere}). 
Since we have 
\begin{displaymath}
\gamma_{n+1}-\gamma_n = \lambda\,\Gamma_{n+1}\times \Gamma_{n},\quad
\left|\gamma_{n+1}-\gamma_n\right|=
|\lambda|\,|\Gamma_{n+1}|\,|\Gamma_n|\,\sin \xi_{n+1}
=\frac{\sin\xi_{n+1}}{|\lambda|},
\end{displaymath}
the tangent vector $T_n$ of $\gamma_n$ is given by
\begin{equation}
 T_n = \frac{\gamma_{n+1}-\gamma_n}{|\gamma_{n+1}-\gamma_n|}
=\frac{\lambda|\lambda|}{\sin\xi_{n+1}}\Gamma_{n+1}\times \Gamma_{n}.
\end{equation}
The binormal vector $B_n=\frac{T_{n-1}\times T_n}{|T_{n-1}\times T_n|}$ of $\gamma_n$ is computed as follows.
First, we have
\begin{equation}
 T_{n-1}\times T_n = \frac{\lambda^4}{\sin\xi_{n+1}\sin\xi_n}
\left\{(\Gamma_{n}\times \Gamma_{n-1})\times(\Gamma_{n+1}\times \Gamma_{n})\right\}.
\end{equation}
Since we have
\begin{displaymath}
 (\Gamma_{n}\times \Gamma_{n-1})\times(\Gamma_{n+1}\times \Gamma_{n})
=\langle \Gamma_{n}\times \Gamma_{n-1},\Gamma_{n}\rangle\Gamma_{n+1}
- \langle \Gamma_{n}\times \Gamma_{n-1},\Gamma_{n+1}\rangle\Gamma_{n}
=-\langle \Gamma_{n}\times \Gamma_{n-1},\Gamma_{n+1}\rangle\Gamma_{n},
\end{displaymath}
we find that $B_n$ is proportional to $\Gamma_n$. Noticing that $\sin\xi_n>0$ for all $n$, we get by normalization
\begin{equation}
 B_n = s_n\,|\lambda|\,\Gamma_n,\quad
s_n = \sgn(\langle \Gamma_{n-1}\times \Gamma_{n},\Gamma_{n+1}\rangle).
\end{equation}
Moreover, the principal normal vector $N_n$ is computed as
\begin{align}
 N_n &= 
B_n\times T_n = 
\frac{\lambda^3s_n}{\sin\xi_{n+1}}
\Gamma_n\times(\Gamma_{n+1}\times \Gamma_{n})
=\frac{\lambda^3s_n}{\sin\xi_{n+1}}
\left(\langle \Gamma_n,\Gamma_n\rangle\Gamma_{n+1} - \langle \Gamma_n,\Gamma_{n+1}\rangle\Gamma_{n}\right)
\nonumber\\
& =\frac{\lambda s_n}{\sin\xi_{n+1}}
\left(\Gamma_{n+1} - \cos\xi_{n+1}\Gamma_{n}\right).
\end{align}
Therefore, from 
\begin{equation}
 \sin\nu_{n+1}=\langle B_{n+1},N_n\rangle
=\langle  |\lambda|\,s_{n+1}\,\Gamma_{n+1},
\frac{\lambda s_n}{\sin\xi_{n+1}}
\left(\Gamma_{n+1} - \cos\xi_{n+1}\Gamma_{n}\right)\rangle
=\frac{|\lambda|}{\lambda}s_ns_{n+1}\sin\xi_{n+1},
\end{equation}
the torsion $\frac{\sin\nu_{n+1}}{|\gamma_{n+1}-\gamma_n|}$ is given by
\begin{equation}
\frac{\sin\nu_{n+1}}{|\gamma_{n+1}-\gamma_n|}=s_{n}s_{n+1}\lambda
\end{equation}
Therefore, if $s_n$ is constant with respect to $n$, $\gamma$ is a discrete space curve with the
constant torsion $\lambda$.
 \end{proof}
%
For the curves on sphere ,it is easy to lift the formulation of the deformation for the Frenet frame
to that for the curve $\Gamma_n$. We define the Frenet frame $F_n$ of the curve $\Gamma_n$
by (see Fig. \ref{fig:sphere})
\begin{equation}\label{A:Frenet}
 F_n = \left[
\frac{\lambda^2}{\sin\xi_{n+1}}\Gamma_n\times \Gamma_{n+1},
\frac{\lambda}{\sin\xi_{n+1}}\left(\Gamma_{n+1}-\cos\xi_{n+1}\Gamma_n\right),
-\lambda\Gamma_n
\right]=[\hat{T}_n,\hat{N}_n,\hat{B}_n].
\end{equation}
Then noticing that $\hat N_n$, $\hat B_n$ and $\hat B_{n+1}$ are coplanar, we find that there
exists a function $\hat\kappa_n$ such that $F_n$ satisfies the Frenet-Serret formula
\begin{equation}\label{A:Frenet-Serret}
 F_{n+1}=F_{n}\,R_1(\xi_{n+1})R_3(\hat\kappa_{n+1}).
\end{equation}
%
 \begin{figure}[ht]
\begin{center}
  \includegraphics[scale=0.5]{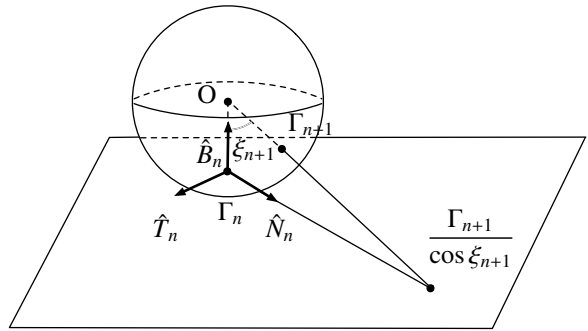}
\caption{The Frenet frame for the curves on sphere.}\label{fig:sphere}
\end{center}
 \end{figure}

For example, let $\xi_n=\xi$ and $a>0$ be constants, and we define the deformation
of the Frenet frame by
\begin{equation}\label{A:M}
\dot{F}_{n} = F_n M_n,\quad 
M_n=\frac{1}{a}
 \left[\begin{array}{ccc}\smallskip
0&-\cos\xi\tan\frac{\hat\kappa_n}{2} - \tan\frac{\hat\kappa_{n+1}}{2} & -\sin\xi\tan\frac{\hat\kappa_n}{2}\\
\smallskip
\cos\xi\tan\frac{\hat\kappa_n}{2} + \tan\frac{\hat\kappa_{n+1}}{2} & 0 & -\sin \xi\\
\sin\xi\tan\frac{\hat\kappa_n}{2} & \sin \xi & 0
	    \end{array}\right],
\end{equation}
then it follows from the compatibility condition of \eqref{A:Frenet-Serret} and \eqref{A:M} the
semi-discrete mKdV equation for $\hat\kappa_n$
\begin{equation}
\dot{\hat\kappa}_n
=\label{A:sdmkdv}
\dfrac{\,1\,}{a}
\left(\tan \dfrac{\hat\kappa_{n + 1}}{2}
- \tan \dfrac{\hat\kappa_{n - 1}}{2}\right).
\end{equation}
To derive the deformation for the curve $\Gamma_n$, since we have $\hat
B_n=-\lambda\Gamma_n$ from \eqref{A:Frenet}, the third column of the matrices in both sides of \eqref{A:M} immediately yields
\begin{equation}
 \dot\Gamma_n = \frac{1}{a\lambda}\left(\sin\xi\tan\frac{\hat\kappa_n}{2}\,\hat T_n
+\sin\xi\,\hat N_n\right).
\end{equation}
This is the deformation of the discrete curves on sphere which is analogous to the deformation of
the discrete space curves discussed in Section \ref{section:semi-discrete}. In order to transform it
to the deformation of the curve $\gamma_n$ in $\R^3$ by using the correspondence in
\eqref{A:def_gamma_n}, it is necessary to perform the summation, which is not a trivial procedure.


\end{document}